\documentclass[journal]{IEEEtran}

\ifCLASSINFOpdf
   \usepackage[pdftex]{graphicx}
 \else
\fi

\usepackage{amsthm}

\usepackage{amssymb}
\usepackage{amsmath}
\usepackage{bm}
\usepackage{cite}
\usepackage{amsmath,amssymb,amsfonts}
\usepackage{algorithmic}
\usepackage{graphicx}
\usepackage{textcomp}
\usepackage{multirow}
\usepackage{caption}
\usepackage{subcaption}
\usepackage{gensymb}
\usepackage{enumitem}
\usepackage[linesnumbered,ruled,vlined]{algorithm2e}
\usepackage{xcolor}
\usepackage[nolist,printonlyused]{acronym}
\usepackage{comment}

\newtheorem{theorem}{Theorem}
\newtheorem{lemma}{Lemma}
\newtheorem{corollary}{Corollary}[lemma]

\begin{acronym}
  \acro{2G}{Second Generation}
  \acro{3G}{3$^\text{rd}$~Generation}
  \acro{3GPP}{3$^\text{rd}$~Generation Partnership Project}
  \acro{4G}{4$^\text{th}$~Generation}
  \acro{5G}{5$^\text{th}$~Generation}
  \acro{6G}{6$^\text{th}$~Generation}
  \acro{AA}{Antenna Array}
  \acro{ABF}{analog beamforming}
  \acro{AC}{Admission Control}
  \acro{ACF}{autocorrelation function}
  \acro{AD}{Attack-Decay}
  \acro{ADC}{analog-to-digital converter}
  \acro{ADSL}{Asymmetric Digital Subscriber Line}
	\acro{AHW}{Alternate Hop-and-Wait}
  \acro{AMC}{Adaptive Modulation and Coding}
	\acro{AP}{Access Point}
  \acro{APA}{Adaptive Power Allocation}
  \acro{AoA}{angle of arrival}
  \acro{AoD}{angle of departure}
  \acro{ARMA}{Autoregressive Moving Average}
  \acro{ATES}{Adaptive Throughput-based Efficiency-Satisfaction Trade-Off}
  \acro{AWGN}{additive white Gaussian noise}
  \acro{BB}{Branch and Bound}
  \acro{BD}{Block Diagonalization}
  \acro{BER}{bit error rate}
  \acro{BF}{Best Fit}
  \acro{BLER}{BLock Error Rate}
  \acro{BPC}{Binary power control}
  \acro{BPSK}{Binary Phase-Shift Keying}
  \acro{BPA}{Best \ac{PDPR} Algorithm}
  \acro{BRA}{Balanced Random Allocation}
  \acro{BS}{base station}
  \acro{CAP}{Combinatorial Allocation Problem}
  \acro{CAPEX}{Capital Expenditure}
  \acro{CBF}{Coordinated Beamforming}
  \acro{CBR}{Constant Bit Rate}
  \acro{CBS}{Class Based Scheduling}
  \acro{CC}{Congestion Control}
  \acro{CDF}{cumulative distribution function}
  \acro{CDMA}{Code-Division Multiple Access}
  \acro{CL}{Closed Loop}
  \acro{CLPC}{Closed Loop Power Control}
  \acro{CNR}{Channel-to-Noise Ratio}
  \acro{CPA}{Cellular Protection Algorithm}
  \acro{CPICH}{Common Pilot Channel}
  \acro{CoMP}{Coordinated Multi-Point}
  \acro{CQI}{Channel Quality Indicator}
  \acro{CRM}{Constrained Rate Maximization}
	\acro{CRN}{Cognitive Radio Network}
  \acro{CS}{compressed sensing}
  \acro{CSI}{channel state information}
  \acro{CSIR}{channel state information at the receiver}
  \acro{CSIT}{channel state information at the transmitter}
  \acro{CRLB}{Cramer-Rao lower bound}
  \acro{CUE}{cellular user equipment}
  \acro{D2D}{device-to-device}
  \acro{DCA}{Dynamic Channel Allocation}
  \acro{DE}{Differential Evolution}
  \acro{DFT}{discrete Fourier transform}
  \acro{DIST}{Distance}
  \acro{DL}{downlink}
  \acro{DMA}{Double Moving Average}
	\acro{DMRS}{Demodulation Reference Signal}
  \acro{D2DM}{D2D Mode}
  \acro{DMS}{D2D Mode Selection}
  \acro{DoA}{direction of arrival}
  \acro{DPC}{Dirty Paper Coding}
  \acro{DRA}{Dynamic Resource Assignment}
  \acro{DSA}{Dynamic Spectrum Access}
  \acro{DSM}{Delay-based Satisfaction Maximization}
  \acro{ECC}{Electronic Communications Committee}
  \acro{EFLC}{Error Feedback Based Load Control}
  \acro{EI}{Efficiency Indicator}
  \acro{eNB}{Evolved Node B}
  \acro{EPA}{Equal Power Allocation}
  \acro{EPC}{Evolved Packet Core}
  \acro{EPS}{Evolved Packet System}
  \acro{E-UTRAN}{Evolved Universal Terrestrial Radio Access Network}
  \acro{ES}{Exhaustive Search}
  \acro{FDD}{frequency division duplexing}
  \acro{FDM}{Frequency Division Multiplexing}
  \acro{FER}{Frame Erasure Rate}
  \acro{FF}{Fast Fading}
  \acro{FSB}{Fixed Switched Beamforming}
  \acro{FST}{Fixed SNR Target}
  \acro{FTP}{File Transfer Protocol}
  \acro{GA}{Genetic Algorithm}
  \acro{GBR}{Guaranteed Bit Rate}
  \acro{GLR}{Gain to Leakage Ratio}
  \acro{GOS}{Generated Orthogonal Sequence}
  \acro{GPL}{GNU General Public License}
  \acro{GRP}{Grouping}
  \acro{HARQ}{Hybrid Automatic Repeat Request}
  \acro{HBF}{hybrid beamforming}
  \acro{HMS}{Harmonic Mode Selection}
  \acro{HOL}{Head Of Line}
  \acro{HSDPA}{High-Speed Downlink Packet Access}
  \acro{HSPA}{High Speed Packet Access}
  \acro{HTTP}{HyperText Transfer Protocol}
  \acro{ICMP}{Internet Control Message Protocol}
  \acro{ICI}{Intercell Interference}
  \acro{ID}{Identification}
  \acro{IDFT}{inverse discrete Fourier transform}
  \acro{IETF}{Internet Engineering Task Force}
  \acro{ILP}{Integer Linear Program}
  \acro{JRAPAP}{Joint RB Assignment and Power Allocation Problem}
  \acro{UID}{Unique Identification}
  \acro{IID}{Independent and Identically Distributed}
  \acro{IIR}{Infinite Impulse Response}
  \acro{ILP}{Integer Linear Problem}
  \acro{IMT}{International Mobile Telecommunications}
  \acro{INV}{Inverted Norm-based Grouping}
	\acro{IoT}{Internet of Things}
  \acro{IP}{Internet Protocol}
  \acro{IPv6}{Internet Protocol Version 6}
  \acro{ISD}{Inter-Site Distance}
  \acro{ISI}{Inter Symbol Interference}
  \acro{ITU}{International Telecommunication Union}
  \acro{JOAS}{Joint Opportunistic Assignment and Scheduling}
  \acro{JOS}{Joint Opportunistic Scheduling}
  \acro{JP}{Joint Processing}
	\acro{JS}{Jump-Stay}
  \acro{KKT}{Karush-Kuhn-Tucker}
  \acro{L3}{Layer-3}
  \acro{LAC}{Link Admission Control}
  \acro{LA}{Link Adaptation}
  \acro{LC}{Load Control}
  \acro{LOS}{Line of Sight}
  \acro{LP}{Linear Programming}
  \acro{LS}{least squares}
  \acro{LTE}{Long Term Evolution}
  \acro{LTE-A}{LTE-Advanced}
  \acro{LTE-Advanced}{Long Term Evolution Advanced}
  \acro{M2M}{Machine-to-Machine}
  \acro{MAC}{Medium Access Control}
  \acro{MANET}{Mobile Ad hoc Network}
  \acro{MC}{Modular Clock}
  \acro{MCS}{Modulation and Coding Scheme}
  \acro{MDB}{Measured Delay Based}
  \acro{MDI}{Minimum D2D Interference}
  \acro{MF}{Matched Filter}
  \acro{MG}{Maximum Gain}
  \acro{MH}{Multi-Hop}
  \acro{MIMO}{multiple input multiple output}
  \acro{MINLP}{Mixed Integer Nonlinear Programming}
  \acro{MIP}{Mixed Integer Programming}
  \acro{MISO}{Multiple Input Single Output}
  \acro{ML}{maximum likelihood}
  \acro{MLWDF}{Modified Largest Weighted Delay First}
  \acro{MME}{Mobility Management Entity}
  \acro{MMSE}{minimum mean square error}
  \acro{MOS}{Mean Opinion Score}
  \acro{MPF}{Multicarrier Proportional Fair}
  \acro{MRA}{Maximum Rate Allocation}
  \acro{MR}{Maximum Rate}
  \acro{MRC}{Maximum Ratio Combining}
  \acro{MRT}{Maximum Ratio Transmission}
  \acro{MRUS}{Maximum Rate with User Satisfaction}
  \acro{MS}{mobile station}
  \acro{MSE}{mean squared error}
  \acro{MSI}{Multi-Stream Interference}
  \acro{MSR}{main-to-secondary-paths power ratio}
  \acro{MTC}{Machine-Type Communication}
  \acro{MTSI}{Multimedia Telephony Services over IMS}
  \acro{MTSM}{Modified Throughput-based Satisfaction Maximization}
  \acro{MU-MIMO}{multiuser multiple input multiple output}
  \acro{MU}{multi-user}
  \acro{NAS}{Non-Access Stratum}
  \acro{NB}{Node B}
  \acro{NE}{Nash equilibrium}
	\acro{NCL}{Neighbor Cell List}
  \acro{NLP}{Nonlinear Programming}
  \acro{NLOS}{Non-Line of Sight}
  \acro{NMSE}{normalized mean square error}
  \acro{NORM}{Normalized Projection-based Grouping}
  \acro{NN}{neural network}
  \acro{NP}{Non-Polynomial Time}
  \acro{NRT}{Non-Real Time}
  \acro{NSPS}{National Security and Public Safety Services}
  \acro{O2I}{Outdoor to Indoor}
  \acro{OFDMA}{orthogonal frequency division multiple access}
  \acro{OFDM}{orthogonal frequency division multiplexing}
  \acro{OFPC}{Open Loop with Fractional Path Loss Compensation}
	\acro{O2I}{Outdoor-to-Indoor}
  \acro{OL}{Open Loop}
  \acro{OLPC}{Open-Loop Power Control}
  \acro{OL-PC}{Open-Loop Power Control}
  \acro{OMP}{orthogonal matching pursuit}
  \acro{OPEX}{Operational Expenditure}
  \acro{ORB}{Orthogonal Random Beamforming}
  \acro{JO-PF}{Joint Opportunistic Proportional Fair}
  \acro{OSI}{Open Systems Interconnection}
  \acro{PAIR}{D2D Pair Gain-based Grouping}
  \acro{PAPR}{Peak-to-Average Power Ratio}
  \acro{P2P}{Peer-to-Peer}
  \acro{PC}{Power Control}
  \acro{PCI}{Physical Cell ID}
  \acro{PDF}{probability density function}
  \acro{PDPR}{pilot-to-data power ratio}
  \acro{PER}{Packet Error Rate}
  \acro{PF}{Proportional Fair}
  \acro{P-GW}{Packet Data Network Gateway}
  \acro{PL}{Pathloss}
  \acro{PPR}{pilot power ratio}
  \acro{PRB}{physical resource block}
  \acro{PROJ}{Projection-based Grouping}
  \acro{ProSe}{Proximity Services}
  \acro{PS}{Packet Scheduling}
  \acro{PSAM}{pilot symbol assisted modulation}
  \acro{PSO}{Particle Swarm Optimization}
  \acro{PZF}{Projected Zero-Forcing}
  \acro{QAM}{Quadrature Amplitude Modulation}
  \acro{QoS}{Quality of Service}
  \acro{QPSK}{Quadri-Phase Shift Keying}
  \acro{RAISES}{Reallocation-based Assignment for Improved Spectral Efficiency and Satisfaction}
  \acro{RAN}{Radio Access Network}
  \acro{RA}{Resource Allocation}
  \acro{RAT}{Radio Access Technology}
  \acro{RATE}{Rate-based}
  \acro{RB}{resource block}
  \acro{RBG}{Resource Block Group}
  \acro{REF}{Reference Grouping}
  \acro{RF}{radio-frequency}
  \acro{RLC}{Radio Link Control}
  \acro{RM}{Rate Maximization}
  \acro{RMSE}{root mean squared error}
  \acro{RNC}{Radio Network Controller}
  \acro{RND}{Random Grouping}
  \acro{RRA}{Radio Resource Allocation}
  \acro{RRM}{Radio Resource Management}
  \acro{RSCP}{Received Signal Code Power}
  \acro{RSRP}{Reference Signal Receive Power}
  \acro{RSRQ}{Reference Signal Receive Quality}
  \acro{RR}{Round Robin}
  \acro{RRC}{Radio Resource Control}
  \acro{RSSI}{Received Signal Strength Indicator}
  \acro{RT}{Real Time}
  \acro{RU}{Resource Unit}
  \acro{RUNE}{RUdimentary Network Emulator}
  \acro{RV}{Random Variable}
  \acro{SAC}{Session Admission Control}
  \acro{SCM}{Spatial Channel Model}
  \acro{SC-FDMA}{Single Carrier - Frequency Division Multiple Access}
  \acro{SD}{Soft Dropping}
  \acro{S-D}{Source-Destination}
  \acro{SDPC}{Soft Dropping Power Control}
  \acro{SDMA}{Space-Division Multiple Access}
  \acro{SER}{Symbol Error Rate}
  \acro{SES}{Simple Exponential Smoothing}
  \acro{S-GW}{Serving Gateway}
  \acro{SINR}{signal-to-interference-plus-noise ratio}
  \acro{SI}{Satisfaction Indicator}
  \acro{SIC}{successive interference cancellation}
  \acro{SIP}{Session Initiation Protocol}
  \acro{SISO}{single input single output}
  \acro{SIMO}{Single Input Multiple Output}
  \acro{SIR}{signal-to-interference ratio}
  \acro{SLNR}{Signal-to-Leakage-plus-Noise Ratio}
  \acro{SMA}{Simple Moving Average}
  \acro{SNR}{signal-to-noise ratio}
  \acro{SORA}{Satisfaction Oriented Resource Allocation}
  \acro{SORA-NRT}{Satisfaction-Oriented Resource Allocation for Non-Real Time Services}
  \acro{SORA-RT}{Satisfaction-Oriented Resource Allocation for Real Time Services}
  \acro{SPF}{Single-Carrier Proportional Fair}
  \acro{SRA}{Sequential Removal Algorithm}
  \acro{SRS}{Sounding Reference Signal}
  \acro{SSE}{sum of squared errors}
  \acro{SU-MIMO}{Single-User Multiple Input Multiple Output}
  \acro{SU}{Single-User}
  \acro{SVD}{singular value decomposition}
  \acro{TCP}{Transmission Control Protocol}
  \acro{TDD}{time division duplexing}
  \acro{TDMA}{Time Division Multiple Access}
  \acro{TETRA}{Terrestrial Trunked Radio}
  \acro{TP}{Transmit Power}
  \acro{TPC}{Transmit Power Control}
  \acro{TTI}{Transmission Time Interval}
  \acro{TTR}{Time-To-Rendezvous}
  \acro{TSDCE}{transformed spatial domain channel estimation}
  \acro{TSM}{Throughput-based Satisfaction Maximization}
  \acro{TU}{Typical Urban}
  \acro{UE}{User Equipment}
  \acro{UEPS}{Urgency and Efficiency-based Packet Scheduling}
  \acro{UL}{uplink}
  \acro{UMTS}{Universal Mobile Telecommunications System}
  \acro{URI}{Uniform Resource Identifier}
  \acro{URM}{Unconstrained Rate Maximization}
  \acro{UT}{user terminal}
  \acro{VR}{Virtual Resource}
  \acro{VoIP}{Voice over IP}
  \acro{WAN}{Wireless Access Network}
  \acro{WCDMA}{Wideband Code Division Multiple Access}
  \acro{WF}{Water-filling}
  \acro{WiMAX}{Worldwide Interoperability for Microwave Access}
  \acro{WINNER}{Wireless World Initiative New Radio}
  \acro{WLAN}{Wireless Local Area Network}
  \acro{WLS}{weighted least squares}
  \acro{WMPF}{Weighted Multicarrier Proportional Fair}
  \acro{WPF}{Weighted Proportional Fair}
  \acro{WSN}{Wireless Sensor Network}
  \acro{WWW}{World Wide Web}
  \acro{XIXO}{(Single or Multiple) Input (Single or Multiple) Output}
  \acro{ZF}{zero-forcing}
  \acro{ZMCSCG}{Zero Mean Circularly Symmetric Complex Gaussian}
  \acro{FFT}{fast Fourier transform}
  \acro{KF}{Kalman filtering}
\end{acronym}

\begin{document}

\title{Fast Channel Estimation in the Transformed Spatial Domain for Analog Millimeter Wave Systems}

\author{Sandra~Roger,~\IEEEmembership{Senior Member,~IEEE,}
        Maximo~Cobos,~\IEEEmembership{Senior Member,~IEEE,}\\
        Carmen~Botella-Mascarell,~\IEEEmembership{Senior Member,~IEEE,}        
        Gábor~Fodor,~\IEEEmembership{Senior Member,~IEEE}
\thanks{S. Roger, M. Cobos and C. Botella-Mascarell are with the Computer Science Department, Universitat de Val\`encia, Av. de la Universitat s/n, 46100 Burjassot, Spain, e-mail: \{sandra.roger, maximo.cobos, carmen.botella\}@uv.es. G. Fodor is with Ericsson Research and KTH Royal Institute of Technology, Malvinasv 10 16400 Stockholm, Sweden, e-mail: gaborf@kth.se.}
\thanks{This work was partially supported by the Spanish Ministry of Science, Innovation and Universities through grant RYC-2017-22101 and project RTI2018-097045-B-C21 (supported also by ERDF), and by the Generalitat Valenciana through projects GV/2020/046 and AICO/2020/154.}
}

\markboth{IEEE Transactions on Wireless Communications}%
{Roger \MakeLowercase{\textit{et al.}}: }

\IEEEoverridecommandlockouts
\IEEEpubid{\begin{minipage}[t]{\textwidth}\ \\[10pt]
        \centering\footnotesize{\copyright 2021 IEEE.  Personal use of this material is permitted.  Permission from IEEE must be obtained for all other uses, in any current or future media, including reprinting/republishing this material for advertising or promotional purposes, creating new collective works, for resale or redistribution to servers or lists, or reuse of any copyrighted component of this work in other works.}
\end{minipage}}

\maketitle

\begin{abstract}
Fast channel estimation in millimeter-wave (mmWave) systems is a fundamental enabler of high-gain beamforming, which boosts coverage and capacity. The channel estimation stage typically involves an initial beam training process where a subset of the possible beam directions at the transmitter and receiver is scanned along a predefined codebook. Unfortunately, the high number of transmit and receive antennas deployed in mmWave systems increase the complexity of the beam selection and channel estimation tasks. In this work, we tackle the channel estimation problem in analog systems from a different perspective than used by previous works. In particular, we propose to move the channel estimation problem from the angular domain into the transformed spatial domain, in which estimating the angles of arrivals and departures corresponds to estimating the angular frequencies of paths constituting the mmWave channel. The proposed approach, referred to as transformed spatial domain channel estimation (TSDCE) algorithm, exhibits robustness to additive white Gaussian noise by combining low-rank approximations and sample autocorrelation functions for each path in the transformed spatial domain. Numerical results evaluate the mean square error of the channel estimation and the direction of arrival estimation capability. TSDCE significantly reduces the first, while exhibiting a remarkably low computational complexity compared with well-known benchmarking schemes. 
\end{abstract}



\section{Introduction}
\IEEEPARstart{T}{he} unprecedented growth of data traffic driven by the increasing number of mobile broadband subscriptions and the increasing data volume per subscription is fueling the evolution of mobile systems~\cite{Ericsson2020}. To support user data rates of Gbps and to meet the insatiable capacity demands in a commercially viable manner, contiguous bandwidths in the order of GHz are required. While such contiguous bandwidths are hardly available in frequencies below 10 GHz (where spectrum is highly fragmented), large monetizable chunks of unused spectrum resources exist in millimeter-wave (mmWave) bands. Recognizing these business and technology drivers, mmWave communications are adopted in 5G \cite{Roh14,And17} and future 6G systems \cite{Gio20}.

Communications in mmWave bands over relatively long distances present challenges due to unfavourable propagation and atmospheric absorption characteristics, however, they also open up for new solutions \cite{Akd14}. Most importantly, they allow to use advanced adaptive array technologies and thereby to achieve substantial beamforming gains, which boost link budgets. To take advantage of directional communications, mmWave communication systems employ sophisticated beam sweeping, measurement and reporting schemes to constantly monitor the direction of transmission (characterized by the \ac{AoA} and \ac{AoD} of the transmitter (Tx) and receiver (Rx) beams) of each potential link \cite{Gio19b,Enescu20}. In practice, this process is carried out utilizing predefined codebooks of directions (identified by a pair of beamforming and combining vectors) that cover the entire angular space between Tx and Rx nodes. To improve the search efficiency among the candidate beam directions, hierarchical codebooks combining an initial coarse beam search with subsequent finer ones are often used \cite{Alk14,Xia16,Xia18}. 

The differences between sub-10 GHz and mmWave bands regarding propagation conditions and required antenna sizes have positioned channel estimation in mmWave channels as an active research field both in academia and industry. Traditional \ac{MIMO} transceiver architectures enable fully-digital  processing through the allocation of one \ac{RF} chain per antenna, which is difficult to realize at mmWave frequencies when large antenna arrays are used \cite{And17}. To overcome the constraints on the number of \ac{RF} chains, the less costly solution is to perform purely \ac{ABF}, where the processing occurs using a single \ac{RF} chain with networks of phase shifters \cite{Zha16,Rat19,Sin20}. In the first releases of the 3GPP New Radio specifications, \ac{ABF} has been shown to provide high rates to a single user and to effectively combat the high path loss caused by mmWave frequencies \cite{Asplund}. As an alternative, \ac{HBF} architectures divide the precoding/combining between the analog and digital domains, utilizing a number of \ac{RF} chains which is lower than the number of antennas \cite{Mo17:2}. The main advantage of \ac{HBF} architectures over \ac{ABF} ones is the possibility to spatially multiplex users and/or streams with a reduced implementation cost with respect to fully-digital architectures. By considering \ac{ABF} and \ac{HBF}, the sparse nature and parametric structure of the mmWave channel has been widely exploited, often taking into account the constraints of using \ac{ABF} \cite{Alk14, Men15, Mo17:1, Mo17:2, Son19}. Schemes relying on \ac{CS} \cite{Son18, Che17, Son19}, such as the methods based on \ac{OMP} \cite{lee2014, Dua11}, have received significant attention. Additionally, \ac{LS}-based approaches \cite{Son18, Son19} have been proposed for beam alignment in mmWave. Unfortunately, under low \ac{SNR} conditions, the channels recovered by \ac{CS} approaches tend to be overwhelmed by noise, leading to degraded performance \cite{Baj10}. Additionally, the accuracy and complexity are high due to the size of the employed dictionaries. The methods in \cite{Mon15,Yi20,Fan18} estimate the channel in the frequency domain after applying a 2D-\ac{DFT} to an initial channel estimate and estimating the \ac{DFT} peaks through iterative cancellation. The main drawback of the latter methods is that they require a high number of \ac{DFT} points, leading to high complexities.

Closely related to the channel estimation problem and highly relevant for system design is \ac{DoA} estimation and tracking in the presence of mobility, a critical issue in future challenging applications such as vehicle-to-anything communications \cite{Gio20}. Several recent works have developed techniques based on \ac{KF} or particle filters that exploit the correlation structure among subsequent mmWave channel realizations under different assumptions~\cite{Zha16,Va16, Jay17,Pai19,Lim19}. However, tracking methods lose accuracy for increased angle deviations and require the support of an external channel estimator to initialize tracking and to re-estimate when there is a sudden change. Besides, these methods are computationally expensive due to matrix inversions and evaluations of derivatives in all the codebook elements.

More recently, machine learning methods have emerged as a powerful tool for addressing various problems in wireless communications \cite{Qin19}. Focusing on channel estimation, a denoising \ac{NN} is employed in an iterative channel estimation scheme in \cite{He18}, while other deep learning architectures have been proposed to reduce the necessary CSI feedback overhead in massive \ac{MIMO} systems \cite{Wen18, Wan20}. Other approaches have been proposed using convolutional \acp{NN} \cite{Don19} or Bayesian learning \cite{Che19}. Most approaches in this family consider a supervised learning framework, where large and carefully labeled datasets are needed for the specific task to be addressed, which renders these schemes difficult to generalize.

In this paper, we propose a novel approach to mmWave channel estimation through \ac{ABF}, which we refer to as \ac{TSDCE}, based on the idea that the mmWave propagation environment between a Tx-Rx pair can be suitably characterized by an observation matrix capturing the channel characteristics over the Tx-Rx codebooks. By selecting an appropriately ordered codebook of \ac{RF} beamforming vectors, the observation matrix corresponds to the 2D-\ac{DFT} of a sum of complex sinusoids in \ac{AWGN} -- referred to as the transformed spatial domain -- with each such sinusoid characterizing an angular component of the multipath channel between the Tx and Rx nodes. The key aspect is that, recognizing such interpretation, a submatrix of the spatial domain observation corresponds to a noisy version of the channel. Interestingly, this noisy version can be more reliably estimated by identifying the spatial frequencies constituting each path, with a direct correspondence to their associated \ac{AoA} and \ac{AoD}. Specifically, the key contributions of this work are as follows:
\begin{itemize}
\item Lemma 1, Lemma 2 and Theorem 1, which together illustrate the foundations and motivation to estimate the mmWave channel in the transformed spatial domain.
\item Algorithm 1, which describes the steps of the proposed TSDCE method.
\item Lemma 3, Lemma 4 and Theorem 2, which provide the upper and lower bound analyses of the method.
\item Performance and complexity analysis of the proposed algorithm. Comparison with several baseline mmWave channel estimation schemes.
\end{itemize}

Our analysis and numerical results indicate that there are several advantages of treating the mmWave channel estimation problem in the transformed spatial domain. First, in contrast to DFT-based benchmarks, the performance of the method does not saturate at high \acp{SNR} and approximates the \ac{CRLB}. Second, its computational complexity is remarkably smaller than other widely used approaches. Finally, the presented scheme is independent of angular deviations and does not rely on an initial channel estimation, as opposed to \ac{KF}-based tracking approaches, which makes it particularly well-suited to typical mmWave environments, in which abrupt changes in the channel gains due to sudden blockages often occur.

The rest of the paper is structured as follows. Section~\ref{Sec:Model} and Section~\ref{Sec:Spatial} describe the system model and the channel estimation problem in the transformed spatial domain, respectively. Section~\ref{Sec:ChannelEst} describes the proposed \ac{TSDCE}, while Section~\ref{sec:bounds} derives the upper and lower performance bounds of the proposed algorithm. Section~\ref{Sec:Complexity} analyzes its performance and complexity. Finally, Section~\ref{Sec:Conc} summarizes the main insights and concludes the paper.

\emph{Notations}: Bold uppercase $\mathbf{A}$ denotes a matrix and bold lowercase $\mathbf{a}$ denotes a column vector. Superscripts $^*$, $^T$, $^H$ and $^{-1}$ denote conjugate, transpose, conjugate transpose and inverse of a matrix, respectively. $\mathrm{vec}(\mathbf{A})$ is a vector obtained through the vectorization of matrix $\mathbf{A}$. $\mathbf{I}_N$ denotes the $N\times N$ identity matrix and $\bm{1}_N$ and $\bm{0}_N$ stand for all-ones and all-zeros $N$-length column vectors, respectively. The symbols $\otimes$ and $\odot$ indicate Kronecker and Hadamard products, respectively. Operators $\mathrm{DFT}_{\mathrm{2D}}\left\lbrace \cdot \right\rbrace$ and $\mathrm{IDFT}_{\mathrm{2D}}\left\lbrace \cdot \right\rbrace$ perform the \ac{DFT} and the \ac{IDFT} two-dimensional operations.  $[\mathbf{A}]_{q,p}$ is the $(q,p)$-th entry of $\mathbf{A}$.  $\left\| \mathbf{A} \right\|_{F}$ is the Frobenius norm. The magnitude and phase of a complex number are denoted by $|\cdot|$ and $\angle(\cdot)$, respectively. $\mathbb{E}\left\{\cdot\right\}$ is the expectation operator and $\mathrm{rank}(\cdot)$ stands for the matrix rank. $\lambda_i(\mathbf{A})$ denotes the $i$-th magnitude-descendent eigenvalue of $\mathbf{A}$. $\mathcal{CN}(m,\sigma^{2})$ is a complex Gaussian random variable with mean $m$ and variance $\sigma^2$. Finally, $\mathbb{C}$ and $\mathbb{R}^{+}$ denote the set of complex and positive real numbers, respectively, while $\mathrm{Re}\left\{\cdot\right\}$ refers to real part of a complex number.

\section{System model} \label{Sec:Model}
In this section, we introduce the system model for mmWave communications and the procedure for conventional codebook-based training to construct the observation matrix.

\subsection{Millimeter Wave Channel and Signal Model}
Let us consider a single-user mmWave geometric channel where the Tx and Rx are both equipped with uniform linear arrays with $n_t$ and $n_r$ antennas, respectively. As in \cite{Aya12,Zha16,Alk14}, the channel is characterized by $L$ scatterers, each one contributing a single propagation path between the Tx and Rx. Defining by $\alpha_l$ the complex channel coefficient affecting the $l$-th path, $l=1,\ldots,L$, and by $\psi_l$ and $\phi_l$ the \ac{AoA} and \ac{AoD} of the $l$-th path, respectively, the channel model depends on the parameter vector $\bm{\theta}~\triangleq~[|\alpha_1|,\angle{\alpha_1},\phi_1,\psi_1,\dots,|\alpha_L|,\angle{\alpha_L},\phi_L,\psi_L]^{T}$. The parametric channel model $\mathbf{H}(\bm{\theta})~\in~\mathbb{C}^{n_r \times n_t}$ is then defined by
\begin{equation}
    \mathbf{H}(\bm{\theta}) = \sqrt{n_t n_r}\sum_{l=1}^{L}\alpha_l \mathbf{a}_r(\psi_l)\mathbf{a}_t^H(\phi_l).\label{eq:channel}
\end{equation}
Without loss of generality, we further assume that the average power gain is equally balanced among the $L$ paths, so that the complex channel coefficients are modeled as independent identically distributed (i.i.d.) random variables with distribution $\alpha_l~\sim~ \mathcal{CN}(0,\sigma^{2}_{\alpha}/L)$. \ac{AoA} and \ac{AoD} are modeled as uniformly distributed random variables $\psi_l,\phi_l \in [0,2\pi]$. By assuming that the antenna separation is one half of the system operating wavelength, the antenna array responses at the Tx and Rx can be respectively expressed as
\begin{eqnarray}
    \mathbf{a}_t(\phi_l) = \frac{1}{\sqrt{n_t}}[1,\,e^{-j\pi \cos\phi_l},\cdots,e^{-j\pi (n_t-1)\cos\phi_l}]^T,
    \label{eq:stvectort}
    \\
    \mathbf{a}_r(\psi_l)=\frac{1}{\sqrt{n_r}}[1,e^{-j\pi \cos\psi_l},\cdots,e^{-j\pi (n_r-1)\cos\psi_l}]^T.
    \label{eq:stvectorr}
\end{eqnarray}
Although the the actual number of paths $L$ constituting the channel may be unknown a priori, measurements at mmWave have demonstrated that the channel at these frequencies is highly sparse, meaning that the value of $L$ is generally low \cite{Akd14}. Note that the dependence of Eq.~(\ref{eq:channel}) on the parameter vector $\bm{\theta}$ will be omitted in what follows for the sake of notation simplicity. 

We consider a mmWave system using purely \ac{ABF}, where the Tx and Rx antennas are connected to a single \ac{RF} chain through a network of digitally controlled phase-shifters. As in previous works \cite{Zha16}, we assume that the beam search space is represented by a codebook
containing a set of $P$ and $Q$ codewords or directions at the Tx
and Rx side, respectively, leading to quantized angles $\bar{\phi}_p$, $p=0,\,1,\ldots,P-1$ and $\bar{\psi}_q$, $q=0,\, 1,\ldots,Q-1$. Then, a pilot-based training phase is carried out for subsequent channel estimation. More specifically, a pilot symbol $\mathrm{s}$ is transmitted and received through all the possible directions at each side. If the Tx uses the RF beamforming vector $\mathbf{f}_p \in \mathbb{C}^{n_t\times 1}$, and the Rx employs the RF combining vector $\mathbf{w}_q\in \mathbb{C}^{n_r\times 1}$, the resulting signal for each pair of directions $(q,p)$ can be written as
\begin{equation}
    y_{q,p}=\sqrt{\rho}\, \mathbf{w}_q^H\mathbf{H}\mathbf{f}_p\,\mathrm{s} + \mathbf{w}_q^H\mathbf{n},
\end{equation}
where $\rho \in \mathbb{R}^{+}$ is the transmit power, and $\mathbf{n} \sim \mathcal{CN}(0,\bm{\Sigma}_{\mathbf{n}})$ is a complex \ac{AWGN} with covariance $\bm{\Sigma}_{\mathbf{n}}=\sigma^2_{n}\mathbf{I}_{n_r}$. The symbol $\mathrm{s}$ is set to 1 for simplicity in what follows, and the system \ac{SNR} is given by $\frac{\rho}{\sigma_{n}^2}$.

After transmitting the pilot through the $Q\times P$ direction combinations, and letting $\mathbf{W} = [\mathbf{w}_0,\ldots,\mathbf{w}_{Q-1}]\in \mathbb{C}^{n_r\times Q}$ and $\mathbf{F}~=~[\mathbf{f}_0,\ldots,\mathbf{f}_{P-1}]\in\mathbb{C}^{n_t\times P}$, the following observation matrix is obtained
\begin{equation}\label{eq:observation}
    \mathbf{Y} = \sqrt{\rho}\, \mathbf{W}^{H}\mathbf{H}\mathbf{F} + \mathbf{N} =\sqrt{\rho}\,  \mathbf{G}(\bm{\theta}) + \mathbf{N},
\end{equation}
where the noise $\mathbf{N}\in \mathbb{C}^{Q\times P}$ contains i.i.d. $\sim\mathcal{CN}(0,\sigma^2_n)$ elements  and $\mathbf{G} \in \mathbb{C}^{Q\times P}$ encodes channel information $\bm{\theta}$.\footnote{The same observation matrix could be constructed using a mmWave \ac{HBF} architecture with $N_{RF}$ RF chains, where each transmitted pilot could be simultaneously received through $N_{RF}$ out of the $Q$ directions to test, as described by equations (1)-(3) in \cite{lee2014}.} By separating the effect of the different scatterers, $\mathbf{Y}$ can be equivalently written as a sum of path contributions $\mathbf{G}^{(l)}(\bm{\theta}_l)\in \mathbb{C}^{Q\times P}$, each one dependent on a parameter vector $\bm{\theta}_l~=~[|\alpha_l|,\angle{\alpha_l},\phi_l,\,\psi_l]^{T}$
\begin{equation}
    \mathbf{Y} = \sqrt{\rho}\sum_{l = 1}^{L}\mathbf{G}^{(l)}(\bm{\theta}_l) + \mathbf{N}.
    \label{eq:Ymodel1}
\end{equation}

As in previous works \cite{Zha16}, if the beamforming/combining vectors are designed to match the array response, i.e. $\mathbf{f}_p~=~\mathbf{a}_t(\bar{\phi}_p)$ and $\mathbf{w}_q~=~\mathbf{a}_r(\bar{\psi}_q)$, the elements $g_{q,p}^{(l)} = \left[\mathbf{G}^{(l)}(\bm{\theta}_l) \right]_{q,p}$ are given by

\begin{equation}
     g_{q,p}^{(l)}(\bm{\theta}_l) = A_l \frac{1-e^{-j\pi n_r (\cos\psi_l - \cos\bar{\psi}_q)}}{1-e^{-j \pi (\cos\psi_l - \cos\bar{\psi}_q)}}\frac{1-e^{j\pi n_t (\cos\phi_l - \cos\bar{\phi}_p)}}{1-e^{j \pi (\cos\phi_l - \cos\bar{\phi}_p)}},
     \label{eq:gcomponent}
\end{equation}
where $A_l = \frac{\alpha_l}{\sqrt{n_t n_r}}$. Finally, note that since the observation is only sensitive to $\cos\phi_l$ and $\cos\psi_l$, while the actual \ac{AoD}s and \ac{AoA}s cover the range $[0,2\pi]$, the quantized angles $\bar{\phi}_p$ and $\bar{\psi}_q$ only need to consider the range $[0,\pi]$.
\subsection{Least Squares Channel Estimation}
A straightforward \ac{LS} solution for the channel estimation problem can be derived by vectorizing the observation matrix $\mathbf{Y}$ (see Eq.~(\ref{eq:observation})) in the following form
\begin{equation}
    \mathbf{y} = \mathrm{vec}(\mathbf{Y})=\sqrt{\rho}\,\mathbf{Q} \mathrm{vec}(\mathbf{H}) + \mathrm{vec}(\mathbf{N}),
\end{equation}
where $\mathbf{Q} = \mathbf{F}^{T} \otimes \mathbf{W}^{H} \in \mathbb{C}^{QP \times n_t n_r}$. Then, the \ac{LS} estimator can be expressed as
\begin{equation}
    \mathrm{vec}(\mathbf{\hat{H}}_{\mathrm{LS}}) = \frac{1}{\sqrt{\rho}}(\mathbf{Q}^{H}\mathbf{Q})^{-1}\mathbf{Q}^{H}\mathbf{y},
\end{equation}
where $\mathbf{Q}^{H}\mathbf{Q}$ has full rank only for $QP \geq n_t n_r$.

\section{Spatial Domain Interpretation}\label{Sec:Spatial}

The proposed method relies on the fact that, under a proper design of a \ac{DFT}-based codebook for \ac{ABF}, the observation matrix $\mathbf{Y}$ corresponds to the 2D-\ac{DFT} of a sum of windowed complex sinusoids embedded in \ac{AWGN}. This motivates an interpretation of the problem in the transformed spatial domain. This section discusses the selected codebook structure and analyzes the form of the observation matrix in such domain.

\subsection{\ac{DFT}-based Codebook}
Let us analyze further the elements of path components $\mathbf{G}^{(l)}(\bm{\theta}_l)$. Eq.~(\ref{eq:gcomponent}) can be written as
\begin{equation}
    g_{q,p}^{(l)}(\bm{\omega}_l) = A_l\frac{1-e^{-j\left(\omega_q -\omega_{\psi_l} \right)n_r}}{1-e^{-j\left(\omega_q -\omega_{\psi_l} \right)}}\frac{1-e^{-j\left(\omega_p -\omega_{\phi_l} \right)n_t}}{1-e^{-j\left(\omega_p -\omega_{\phi_l} \right)}},
    \label{eq:gfreq}
\end{equation}
where 
\begin{eqnarray}
\omega_q = -\pi \cos(\bar{\psi}_q),\quad &&\omega_p = \pi \cos(\bar{\phi}_p),\\
\omega_{\psi_l} = -\pi \cos(\psi_l), \quad &&\omega_{\phi_l} = \pi \cos(\phi_l).\label{eq:omegaphil}
\end{eqnarray}

It can be shown that Eq.~\eqref{eq:gfreq} corresponds to the 2D-\ac{DFT} with $Q\times P$ bins of a windowed complex sinusoid provided that 
\begin{eqnarray}
    e^{j\omega_q} &=& e^{j\frac{2\pi}{Q}q}, \quad q = 0,1,\dots,Q-1, 
    \label{eq:cond1}
    \\
    e^{j\omega_p} &=& e^{j\frac{2\pi}{P}p}, \quad p = 0,1,\dots,P-1.
    \label{eq:cond2}
\end{eqnarray}

Indeed, when the above relationships hold, it follows
\begin{equation}
\begin{split}
    g_{q,p}^{(l)}(\bm{\omega}_l)  &= \sum_{m = 0}^{n_r-1}\sum_{n=0}^{n_t-1} A_l e^{j\left(\omega_{\psi_l}m + \omega_{\phi_l}n \right)} e^{-j2\pi(\frac{q}{Q}m + \frac{p}{P}n)} \\
    &=  A_l \frac{1-e^{-j\left(\omega_q -\omega_{\psi_l} \right)n_r}}{1-e^{-j\left(\omega_q -\omega_{\psi_l} \right)}} \frac{1-e^{-j\left(\omega_p -\omega_{\phi_l} \right)n_t}}{1-e^{-j\left(\omega_p -\omega_{\phi_l} \right)}},
\end{split}
\end{equation}
which indicates that $g^{(l)}_{q,p}$ is the $(q,p)$ coefficient of the 2D-DFT of $A_l e^{j\left(\omega_{\psi_l}m + \omega_{\phi_l}n \right)}$. The variables $\omega_{\psi_l}, \omega_{\phi_l} \in [-\pi,\pi]$ denote the frequencies of such complex sinusoid in each spatial direction, where the vertical direction is related to the \ac{AoA} and the horizontal direction to the \ac{AoD}. Note that the dependency on $\bm{\theta}_l$ in Eq.~\eqref{eq:gcomponent} has been changed to $\bm{\omega}_l = [|\alpha_l|,\angle{\alpha_l},\omega_{\phi_l},\, \omega_{\psi_l}]^{T}$ in Eq.~\eqref{eq:gfreq} to emphasize the focus on the spatial frequencies rather than on the \ac{AoD}s and \ac{AoA}s.

To satisfy Eqs.~(\ref{eq:cond1}) and (\ref{eq:cond2}), we impose proper conditions on the selected codebook angles as follows
\begin{equation}
    \cos(\bar{\phi}_p) = \mathcal{W}_{[-1,1]}\left(\frac{2p}{P}\right), \quad
    \cos(\bar{\psi}_q) = \mathcal{W}_{[-1,1]}\left(-\frac{2q}{Q}\right),
\end{equation}
where $\mathcal{W}_{[a,b]}(x)\triangleq x - (b-a)\left\lceil \frac{x-b}{b-a} \right\rceil$ is the $[a,\,b]$ wrapping operator with $\left\lceil \cdot \right\rceil$ denoting the ceiling function. The above conditions imply simultaneously a uniform quantization in the range $[-1,\,1]$ for the cosine of the codebook angles and a specific codebook ordering at the Tx and the Rx. 

\subsection{Sinusoidal Path Components}

By considering the above codebook design, the path components 
can be expressed as
\begin{equation}
    \mathbf{G}^{(l)}(\bm{\theta}_l) = \mathrm{DFT}_{\mathrm{2D}}\left\lbrace \mathbf{C}^{(l)}(\bm{\omega}_l)  \right\rbrace,\ \forall l, 
    \label{eq:GDFT}
\end{equation}
where $\mathbf{C}^{(l)}(\bm{\omega}_l)\in\mathbb{C}^{Q\times P}$ are the corresponding spatial domain equivalents, with elements $ c_{m,n}^{(l)}(\bm{\omega}_l) = [\mathbf{C}^{(l)}(\bm{\omega}_l)]_{m,n}$ given by
\begin{equation}
   c_{m,n}^{(l)}(\bm{\omega}_l) = \begin{cases}
    A_le^{j\left(\omega_{\psi_l}m + \omega_{\phi_l}n  \right)},& \text{if } m<n_r,\, n<n_t \\
    0,              & \text{elsewhere}
    \end{cases}.\label{eq:cmn}
\end{equation}

Eq.~\eqref{eq:cmn} clearly reflects the correspondence of $\mathbf{C}^{(l)}(\bm{\omega}_l)$ to a windowed complex sinusoid. In what follows, the dependence on $\bm{\omega}_l$ will be omitted for the sake of notation simplicity. Note that the indices of the elements of $\mathbf{C}^{(l)}$ are denoted as $(m,n)$ to make clearer their correspondence to the vertical and horizontal spatial directions $m~=~0,\dots,Q-1$ and $n~=~ 0,\dots,P-1$, respectively. The windowing effect can be alternatively expressed in matrix notation by defining a binary masking matrix $\mathbf{B}~\in~ \mathbb{N}^{Q \times P}$ applied over a full (non-windowed) cisoid matrix $\mathbf{T}^{(l)}$
\begin{equation}
     \mathbf{C}^{(l)}  =  \left(\alpha_l\mathbf{c}(\omega_{\psi_l})\mathbf{c}(\omega_{\phi_l})^{H}\right)\odot\left( \mathbf{b}_{n_r}\mathbf{b}_{n_t}^{T}\right) 
      = \mathbf{T}^{(l)}\odot \mathbf{B}, 
\end{equation}
where $\mathbf{b}_{n_r} \triangleq [\bm{1}_{n_r}^{T},\, \bm{0}_{Q-n_r}^{T}]^{T}$,
$\mathbf{b}_{n_t} \triangleq [\bm{1}_{n_t}^{T},\, \bm{0}_{P-n_t}^{T}]^{T}$, and
\begin{eqnarray}
&&\mathbf{c}(\omega_{\psi_l}) \triangleq \frac{1}{\sqrt{n_r}} [1,\, e^{j\omega_{\psi_l}},\dots,e^{j(Q-1)\omega_{\psi_l}}]^{T}, \label{eq:cvectors1}\\
&&\mathbf{c}(\omega_{\phi_l}) \triangleq \frac{1}{\sqrt{n_t}} [1,\, e^{j\omega_{\phi_l}},\dots,e^{j(P-1)\omega_{\phi_l}}]^{T}.\label{eq:cvectors2}
\end{eqnarray}

An example of the magnitude of one path component and its spatial equivalent is shown in Fig.~\ref{fig:frequencies}, where both the windowing effect and the spatial sinusoidal pattern are clearly observed.

\begin{figure*}[t]
\begin{center}
\includegraphics[width=0.8\textwidth]{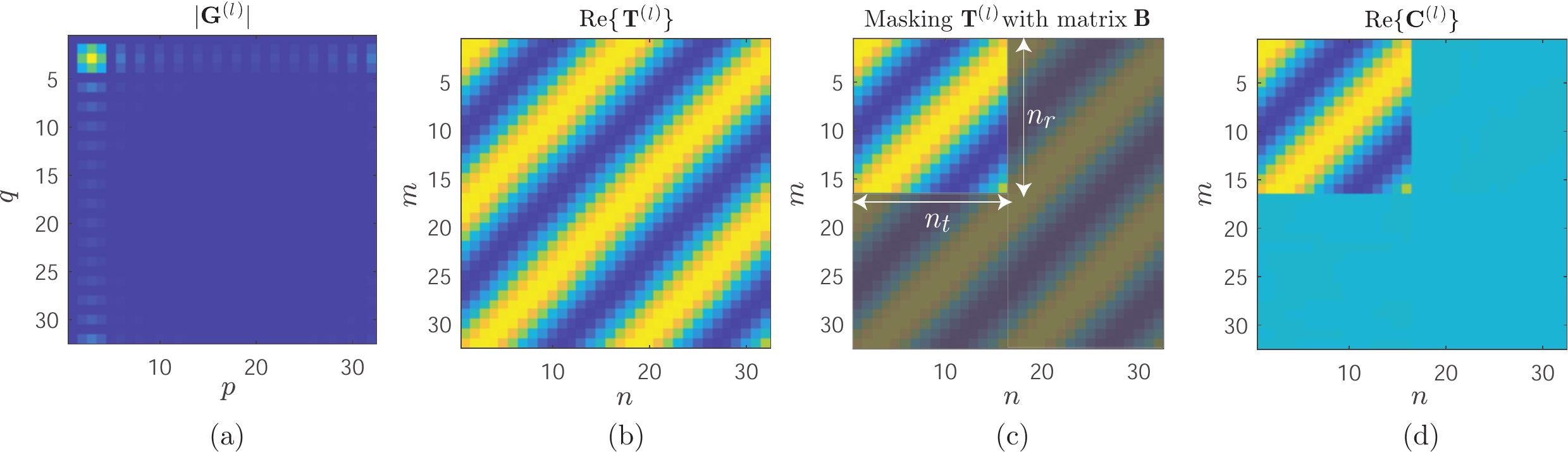}
\caption{{Example illustrating the dual interpretation of one path component in the angular (a) and spatial (d) domains, with $n_t = n_r = 16$ and $P = Q = 32$. The \ac{AoD}/\ac{AoA} are set to match two quantized angles: $\phi_1 = \bar{\phi}_3$, $\psi_1 = \bar{\psi}_3$.} }\label{fig:frequencies} 
\end{center}
\end{figure*}

\subsection{Observation Matrix in the Transformed Spatial Domain}

Let us write the observation matrix $\mathbf{Y}$ in the transformed spatial domain by considering Eq.~(\ref{eq:Ymodel1}) and Eq.~(\ref{eq:GDFT})
\begin{equation}
    \mathbf{D} = \mathrm{IDFT}_{\mathrm{2D}}\left\lbrace \mathbf{Y} \right\rbrace =  \sqrt{\rho}\,\mathbf{C}(\bm{\omega}) + {\mathbf{Z
    }}=\sqrt{\rho} \sum_{l=1}^{L}\mathbf{C}^{(l)} + \mathbf{Z},
    \label{eq:Sspatial}
\end{equation}
where the full parameter vector, equivalent to $\bm{\theta}$, is given by $\bm{\omega} \triangleq [\bm{\omega}_{1}^T,\bm{\omega}_{2}^T,\dots,\bm{\omega}_{L}]^{T}$ and $\mathbf{Z} ~=~\mathrm{IDFT}_{\mathrm{2D}}\lbrace \mathbf{N} \rbrace\in \mathbb{C}^{Q\times P}$ is a noise matrix with i.i.d.
elements corresponding to zero-mean complex Gaussian noise with variance $\sigma_z^2~=~\frac{1}{QP}\sigma^2_n$ \cite{Thi14}.

By considering the effect of the masking matrix $\mathbf{B}$ and its logical negation $\neg{\mathbf{B}}$, the spatial domain observation can be alternatively expressed as
\begin{equation}
    \mathbf{D} =
    \underbrace{
   \sqrt{\rho}\sum_{l=1}^{L}\mathbf{T}^{(l)}\odot \mathbf{B} + \mathbf{Z}\odot \mathbf{B}}_{\mathbf{D}_{\mathbf{C}}} + \underbrace{\mathbf{Z}\odot\mathbf{\neg{B}}}_{\mathbf{D}_{\mathbf{Z}}}.
   \label{eq:S_Sc_Sz}
\end{equation}

The above two differentiated terms allow to write $\mathbf{D}$ as the composition of two non-overlapping parts, defining the full observation as the union of two disjoint sets containing the non-zero elements of $\mathbf{D}_{\mathbf{C}}$ and $\mathbf{D}_{\mathbf{Z}}$
\begin{eqnarray}
\mathcal{D}_{\mathbf{C}} &=& \left\{[\mathbf{D}]_{m,n}: \quad  [\mathbf{B}]_{m,n}=1 \right\}, \\
\mathcal{D}_{\mathbf{Z}} &=& \left\{[\mathbf{D}]_{m,n}: \quad  [\mathbf{B}]_{m,n}=0 \right\},
\end{eqnarray}
i.e. $\mathcal{D} \triangleq \mathcal{D}_{\mathbf{C}} \cup \mathcal{D}_{\mathbf{Z}}$ and $ \mathcal{D}_{\mathbf{C}} \cap \mathcal{D}_{\mathbf{Z}}=\emptyset$.

Fig.~\ref{fig:codebook_size} illustrates the original and transformed observations for three cases (with $L = 1$) corresponding to codebooks of increasing sizes, namely $P=Q=16$, $P=Q=32$ and $P=Q=64$, keeping the number of Tx and Rx antennas $n_t=n_r=16$ and the noise power fixed in all cases. Each column shows the magnitudes of the noiseless observation $\mathbf{G}$, of the noisy observation $\mathbf{Y}$ and of its spatial domain counterpart $\mathbf{D}$, respectively. The last row reflects the location of the above sets within $\mathbf{D}$. In the particular case of a matching number of antennas and codebook size (first column of Fig.~\ref{fig:codebook_size}), the set $\mathcal{D}_{\mathbf{Z}}$ is empty. Note that the signal information concentrates on the $n_r\times n_t$ submatrix from the top-left corner. This submatrix contains the informative part of $\mathbf{D}$, and it is denoted in what follows as  $\mathbf{\bar{D}}_{\mathbf{C}}\in \mathbb{C}^{n_r \times n_t}$. Finally, note that when $P>n_t$ or $Q>n_r$, an estimate of the noise variance can be directly obtained from the elements contained in $\mathcal{D}_{\mathbf{Z}}$ as
\begin{equation}
    \hat{\sigma}^{2}_z = \frac{1}{QP}\hat{\sigma}_n^{2} =  \frac{1}{|\mathcal{D}_{\mathbf{Z}}|}\sum_{d_{m,n} \in \mathcal{D}_{\mathbf{Z}}}|d_{m,n}|^2,
    \label{eq:varnoise1}
\end{equation}
where $|\mathcal{D}_{\mathbf{Z}}|$ denotes the cardinality of the set and $d_{m,n}=[\mathbf{D}]_{m,n}$.

\begin{figure}[t]
\begin{center}
\includegraphics[width =\columnwidth]{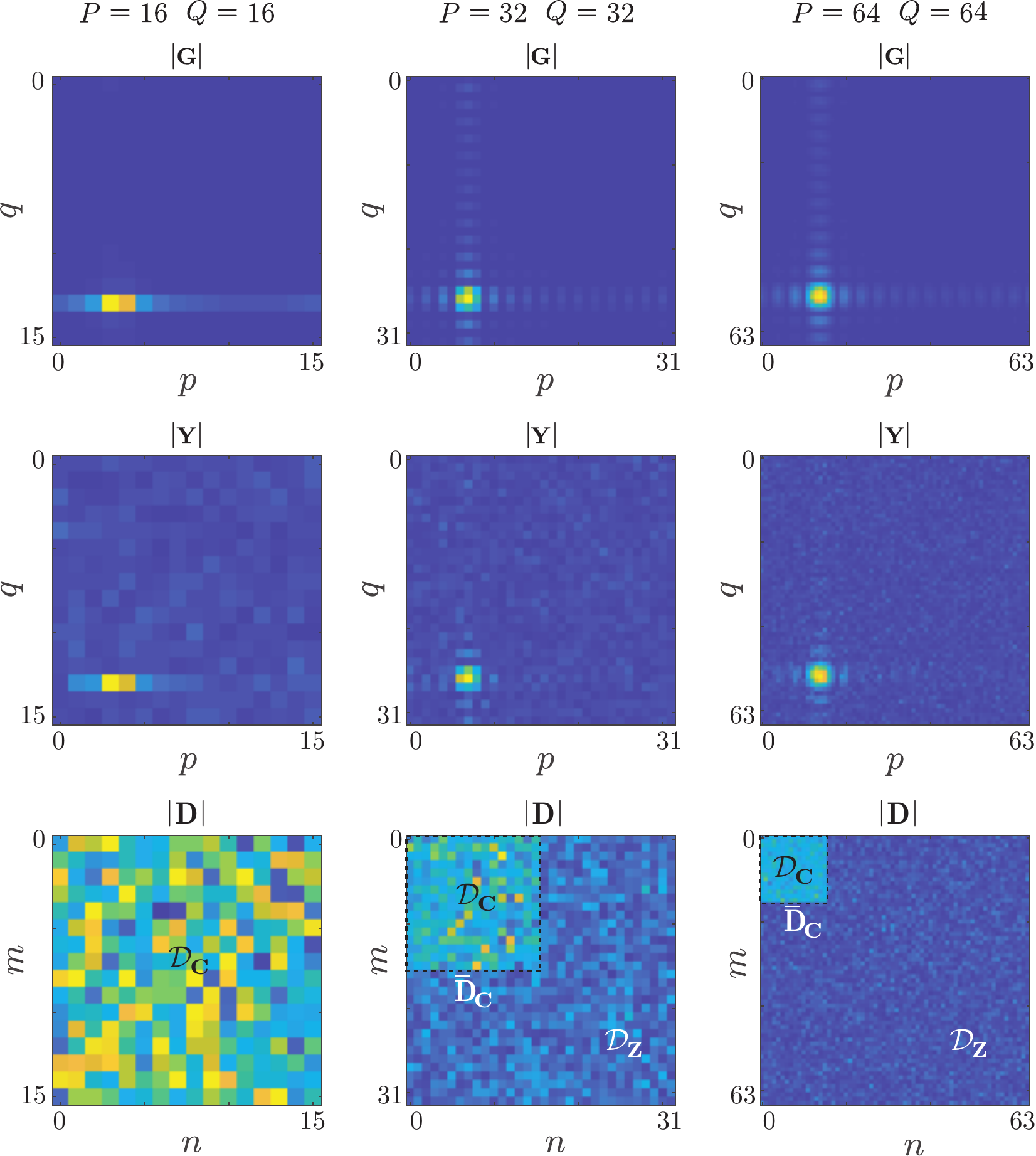}
\caption{{Effect of codebook size for $L=1$ and $n_t = n_r = 16$ with $\mathrm{SNR}=0$ dB. } }\label{fig:codebook_size} 
\end{center}
\end{figure}

\begin{lemma} \label{lem:snrc}
The SNR in the informative part of the spatial domain observation is scaled by a factor $\frac{QP}{n_t n_r}$.
\end{lemma}
\begin{proof}
Let us consider the informative part of the observation matrix $\mathbf{D}$ consisting of the non-zero elements of $\mathbf{D}_{\mathbf{C}}$, denoted as $\mathbf{\bar{D}}_{\mathbf{C}}$, which can be expressed as
\begin{equation}
   \mathbf{\bar{D}}_{\mathbf{C}} = \sqrt{\rho}\,\mathbf{\bar{C}}(\bm{\omega}) + {\mathbf{\bar{Z
    }}} =  \sqrt{\rho}\sum_{l=1}^{L}{\mathbf{\bar{C
    }}}^{(l)} + {\mathbf{\bar{Z
    }}},
    \label{eq:barDc}
\end{equation}
where ${\mathbf{\bar{C}}}^{(l)}\in \mathbb{C}^{n_r \times n_t}$ and $\mathbf{\bar{Z}}\in \mathbb{C}^{n_r \times n_t}$ are sub-matrices of $\mathbf{C}^{(l)}$ and $\mathbf{Z}$ corresponding to $m < n_r$ and $n< n_t$. Thus, the elements $ \bar{d}_{m,n} = [\mathbf{\bar{D}}_{\mathbf{C}}]_{m,n} = [\mathbf{D}]_{m,n}$,  $m=0,\ldots,n_r-1$, $n=0,\ldots,n_t-1$ are formed by both sinusoidal components and noise, i.e.
\begin{equation}
    \bar{d}_{m,n} = \sqrt{\rho}\sum_{l=1}^{L}A_l e^{j\left(\omega_{\psi_l}m + \omega_{\phi_l}n \right)} + z_{m,n},
    \label{eq:dbarnoise}
\end{equation}
where $z_{m,n} = [\mathbf{Z}]_{m,n}$. The variance for the elements $\bar{d}_{m,n}$ is given by
\begin{equation}
    \mathrm{var}(\bar{d}_{m,n})  = \rho\, \mathbb{E}\left\{\sum_{l=1}^{L}|A_l|^2\right\} + \sigma^{2}_z.
    \label{eq:varsc}
\end{equation}
Therefore, the SNR of $\mathbf{\bar{D}}_{\mathbf{C}}$ is given by
\begin{equation}
\begin{split}
   \mathrm{SNR}_{\mathbf{C}}
     &= \frac{\rho\,\mathbb{E}\left\{\sum_{l=1}^{L} |A_l|^2\right\}}{\sigma^2_z} = \frac{\rho}{\sigma_n^2} \frac{QP}{n_t n_r}\sum_{l=1}^{L}\mathbb{E}\left\{|\alpha_l|^2\right\} 
     \\&= \frac{\rho}{\sigma_n^2} \frac{QP}{n_t n_r} \sum_{l=1}^{L}\frac{\sigma^{2}_{\alpha}}{L} =  \frac{\rho}{\sigma_n^2} \frac{QP}{n_t n_r} \sigma^{2}_{\alpha}.
     \end{split}
\end{equation}

As in previous works \cite{Alk14,lee2014}, by assuming that the average power gain of the paths add up to unit power, i.e.  $\sigma^2_{\alpha}=1$, the following relationship can be established
\begin{equation}
    \mathrm{SNR}_{\mathbf{C}}
    =\frac{\rho}{\sigma_n^2} \frac{QP}{n_t n_r}=  \mathrm{SNR} \cdot \frac{QP}{n_t n_r},
    \label{eq:snrgain}
\end{equation}
which evidences a SNR gain when $QP > n_t n_r$. 
\end{proof}

Lemma \ref{lem:snrc} establishes as well an interesting relationship to intuitively understand the effect of the codebook size and the number of antenna elements on the performance by looking into the way the information is distributed within the transformed spatial domain, as illustrated by  Fig.~\ref{fig:codebook_size}. 
An important remark is that, as $A_l$ does not depend on $Q$ and $P$, while the magnitude of the target signal remains constant, the noise power decreases as the codebook size gets larger, since it gets distributed over the whole space. This leads to the SNR gain for the target signal shown by Eq.~(\ref{eq:snrgain}).

\begin{lemma} \label{lem:channel}
Sub-matrix $\mathbf{\bar{D}}_{\mathbf{C}}$ constitutes a scaled noisy observation of the channel matrix $\mathbf{H}$, with noise variance $1/\mathrm{SNR}_{\mathbf{C}}$.
\end{lemma}
\begin{proof}
Taking into account Eqs.~(\ref{eq:stvectort}) and (\ref{eq:stvectorr}) together with the relationships of Eq.~(\ref{eq:omegaphil}), the elements $h_{m,n} = [\mathbf{H}]_{m,n}$ can be written as
\begin{equation}
    h_{m,n} = \sum_{l=1}^{L}\alpha_l e^{j\left(\omega_{\psi_l}m + \omega_{\phi_l}n \right)},
\end{equation}
indicating that the channel consists as well of a sum of $L$ complex sinusoids. By properly scaling the elements $\bar{d}_{m,n}$ with the factor $\sqrt{n_t n_r / \rho}$, we obtain
\begin{equation}
\begin{split}
    \sqrt{\frac{n_t n_r}{\rho}}\bar{d}_{m,n}  =&= \sum_{l=1}^{L}\alpha_l e^{j\left(\omega_{\psi_l}i + \omega_{\phi_l}k \right)} + \sqrt{\frac{n_t n_r}{\rho}}z_{m,n}
     \\&= h_{m,n} + \check{z}_{m,n},
    \label{eq:Hest1}
    \end{split}
\end{equation}
where the new noise variance of $\check{z}_{m,n}$ is
\begin{equation}
    \sigma^2_{\check{z}} = \frac{n_t n_r}{\rho}\sigma^2_z = \frac{n_t n_r}{\rho}\frac{\sigma^2_n}{QP} = \frac{1}{\mathrm{SNR}_{\mathbf{C}}}.
\end{equation}
\end{proof}

\begin{corollary}\label{col:ssesc}
An estimator of the channel matrix $\mathbf{H}$ can be obtained from $\mathbf{\bar{D}}_{\mathbf{C}}$, with mean \ac{SSE} equal to $n_t n_r / \mathrm{SNR}_{\mathbf{C}}$.    
\end{corollary}
\begin{proof}
From Lemma \ref{lem:channel}, a straightforward channel estimate is given by
\begin{equation}
    \mathbf{\hat{H}}_{\mathbf{D}} =  \sqrt{\frac{n_t n_r}{\rho}}\mathbf{\bar{D}}_{\mathbf{C}},
    \label{eq:Hdest}
\end{equation}
which, according to Eq.~(\ref{eq:Hest1}), has elements $\hat{h}_{m,n} = h_{m,n} + \check{z}_{m,n}$. Then, the mean \ac{SSE} for such estimator is easily derived as
\begin{equation}
\begin{split}
  & \mathbb{E} \left\{  \left\|\mathbf{\hat{H}}_{\mathbf{D}}-\mathbf{H}\right \|_F^2 \right\} =  \mathbb{E} \left\{ \sum_{m=0}^{n_r-1}\sum_{n=0}^{n_t-1} \left(\hat{h}_{m,n}-h_{m,n}\right)^2 \right\} =\\ &= 
   \sum_{m=0}^{n_r-1}\sum_{n=0}^{n_t-1}\mathbb{E} \left\{\check{z}_{m,n}^2 \right\}
    = n_t n_r  \sigma^2_{\check{z}} = \frac{n_t n_r}{\mathrm{SNR}_\mathbf{C}}.
    \end{split}
  \end{equation}
\end{proof}

\begin{theorem}
The channel estimator $\mathbf{\hat{H}}_{\mathbf{D}}$ provides a performance equivalent to the LS estimator $\mathbf{\hat{H}}_{\mathrm{LS}}$.
\end{theorem}
\begin{proof}
The mean \ac{SSE} of the \ac{LS} estimator is given by
\begin{equation}
     \mathbb{E} \left\{ \left\|\mathbf{\hat{H}}_{\mathrm{LS}}-\mathbf{H}\right \|_F^2 \right\}  = \frac{1}{\rho}  \mathrm{Tr}(\mathbf{Q}^{H}\mathbf{Q})^{-1}\sigma^2_n 
      = \frac{1}{\rho}\sum_{i=1}^{n_t n_r}\frac{\sigma^2_n}{\lambda_i(\mathbf{Q}^{H}\mathbf{Q})},
\end{equation}
where $\lambda_i(\mathbf{Q}^{H}\mathbf{Q})$ are the eigenvalues of $\mathbf{Q}^{H}\mathbf{Q}$. Given the column orthogonality of both $\mathbf{F}$ and $\mathbf{W}$, it follows that $\mathbf{Q}^{H}\mathbf{Q}$ has also orthogonal columns, so that $\mathbf{Q}^{H}\mathbf{Q} = \frac{QP}{n_t n_r} \mathbf{I}$ is a diagonal matrix with eigenvalues $\lambda_i = \frac{QP}{n_t n_r}$. Therefore,
\begin{equation}
     \mathbb{E} \left\{ \left\|\mathbf{\hat{H}}_{\mathrm{LS}}-\mathbf{H}\right \|_F^2 \right\}  = \frac{\sigma^2_n}{\rho} \sum_{i=1}^{n_t n_r} \frac{n_t n_r}{QP} = \frac{n_t n_r}{\mathrm{SNR}_{\mathbf{C}}},
\end{equation}
which is the same result as the one of Corollary~\ref{col:ssesc}.

\end{proof}

\subsection{Multi-path Analysis}
Let us analyze further the structure of the path components within the cropped spatial domain observation $\mathbf{\bar{D}}_{\mathbf{C}}$ (Eq.~\eqref{eq:barDc}). It can be easily shown that each path component in the spatial domain can be expressed as an outer product of vectors similar to those in Eqs.~\eqref{eq:cvectors1} and \eqref{eq:cvectors2}, given by
\begin{eqnarray}
    \mathbf{\bar{C}}^{(l)} &=& \alpha_l \mathbf{\bar{c}}(\omega_{\psi_l})\mathbf{\bar{c}}(\omega_{\phi_l})^{H}, \label{eq:cvectors0}\\
\mathbf{\bar{c}}(\omega_{\psi_l}) &\triangleq& \frac{1}{\sqrt{n_r}} [1,\, e^{j\omega_{\psi_l}},\dots,e^{j(n_r-1)\omega_{\psi_l}}]^{T}, \label{eq:cbvectors1t}\\
\mathbf{\bar{c}}(\omega_{\phi_l}) &\triangleq& \frac{1}{\sqrt{n_t}} [1,\, e^{j\omega_{\phi_l}},\dots,e^{j(n_t-1)\omega_{\phi_l}}]^{T}.
\label{eq:cbvectors2t}
\end{eqnarray}

The above expressions indicate that the noiseless observation can be interpreted as a sum of rank-one components. In general, $\mathrm{rank}(\mathbf{\bar{C}})=L$, since only the case when the set of frequencies in $\bm{\omega}$ contains repeated elements leads to $\mathrm{rank}(\mathbf{\bar{C}}) < L$. To gain further insight, let us apply the \ac{SVD} to matrix $\mathbf{\bar{C}}$
\begin{equation}
    \mathbf{\bar{C}} = \mathbf{U}\mathbf{S}\mathbf{V}^{H},
\end{equation}
where $\mathbf{S} = \mathrm{diag}(s_1,\, s_2,\, \dots,\,s_{n_t})$ is a diagonal matrix containing the magnitude-ordered singular values of $\mathbf{\bar{C}}$, while $\mathbf{U} = [\mathbf{u}_1,\, \mathbf{u}_2,\dots,\, \mathbf{u}_{n_t}] \in \mathbb{C}^{n_r\times n_t}$ and $\mathbf{V}= [\mathbf{v}_1,\, \mathbf{v}_2, \dots,\, \mathbf{v}_{n_t}] \in \mathbb{C}^{n_t\times n_t}$ are orthonormal matrices whose columns are the corresponding left and right singular vectors, respectively, with elements $\mathbf{u}_i= [u_{i_1},\,u_{i_2},\ldots,u_{i_{n_r}}]^{T}$ and $\mathbf{v}_i=[v_{i_1},\,v_{i_2},\ldots,v_{i_{n_r}}]^{T}$.

By using this decomposition, and assuming that $\mathrm{rank}(\mathbf{\bar{C}})=L$, $\mathbf{\bar{C}}$ can be expressed as a sum of rank-one matrices as 
\begin{equation}
    \mathbf{\bar{C}} = \sum_{l=1}^{L} s_l \mathbf{u}_l \mathbf{v}_l^{H}.
\end{equation}

Without loss of generality, we can consider the $L$ path components in the system model to be ordered according to a decreasing power criterion. By comparing each of the rank-one matrices with Eq.~(\ref{eq:cvectors0}) and taking into account that both $\mathbf{\bar{c}}(\omega_{\psi_l})$ and $\mathbf{\bar{c}}(\omega_{\phi_l})$ are unit vectors, the following approximations can be established
\begin{eqnarray}
s_l &\simeq& |\alpha_l|, \label{eq:svd1}\\
\mathbf{u}_l &\simeq& e^{j\angle u_{l_1}}\mathbf{\bar{c}}(\omega_{\psi_l}), \label{eq:svd2}\\
\mathbf{v}_l &\simeq& e^{j\angle v_{l_1}} \mathbf{\bar{c}}(\omega_{\phi_l}),\label{eq:svd3}\\
\angle\left(u_{l_1}\cdot v_{l_1} \right) &\simeq& \angle(\alpha_l).\label{eq:svd4}
\end{eqnarray} 

A relevant fact is that equality in Eqs.~(\ref{eq:svd1}-\ref{eq:svd4}) would only hold in the case where the set of vectors $\mathbf{\bar{c}}(\omega_{\psi_l})$ and  $\mathbf{\bar{c}}(\omega_{\phi_l})$ were also orthogonal. That would be the case when the involved angular frequencies match, respectively, those of a \ac{DFT} of size $n_r$ and $n_t$, i.e.
\begin{eqnarray}
    \forall l \: \exists m \in 0,\dots,n_r-1 \quad &\mathrm{s.t.}& \quad e^{j\omega_{\psi_l}} = e^{j\frac{2\pi}{n_r}m},\\
     \forall l \: \exists n \in 0,\dots,n_t-1 \quad &\mathrm{s.t.}& \quad e^{j\omega_{\phi_l}} = e^{j\frac{2\pi}{n_t}n}.
\end{eqnarray}

In such case, the path components can be exactly recovered from the \ac{SVD}. Therefore, this analysis also shows that the use of a larger number of antennas would increase the chance of having orthogonal paths, as it allows for having a larger number of \ac{DFT} vectors. Obviously, having a small number of paths favors their separability. Despite that, in general, the paths are not completely orthogonal, the mmWave channel has by nature a small number of clusters with unbalanced powers \cite{Akd14}. Therefore, although the rank-one approximation of $\mathbf{\bar{C}}$ relying on the largest singular value may incorporate residuals from the rest of paths, the effect of these is expected to be low, as the simulation results will confirm. This motivates the use of a \ac{SIC} approach in this paper to retrieve the parameters corresponding to multiple paths.

\section{Transformed Spatial Domain Channel Estimation}\label{Sec:ChannelEst}

This section starts with a detailed description of the estimation of the channel parameters corresponding to the dominant path, whose steps constitute the core of each iteration of the multipath TSDCE algorithm (schematically summarized in Fig.~\ref{fig:diagram}). The first step is based on the \ac{SVD} analysis of the observation matrix in the transformed spatial domain. The second step is based on the denoising properties of the sample \ac{ACF}. Sect.~\ref{sec:sic} will address how the subsequent processing stages can be successively applied to retrieve the parameters from the rest of paths.
\begin{figure*}[t!]
\begin{center}
\includegraphics[width =0.8\textwidth]{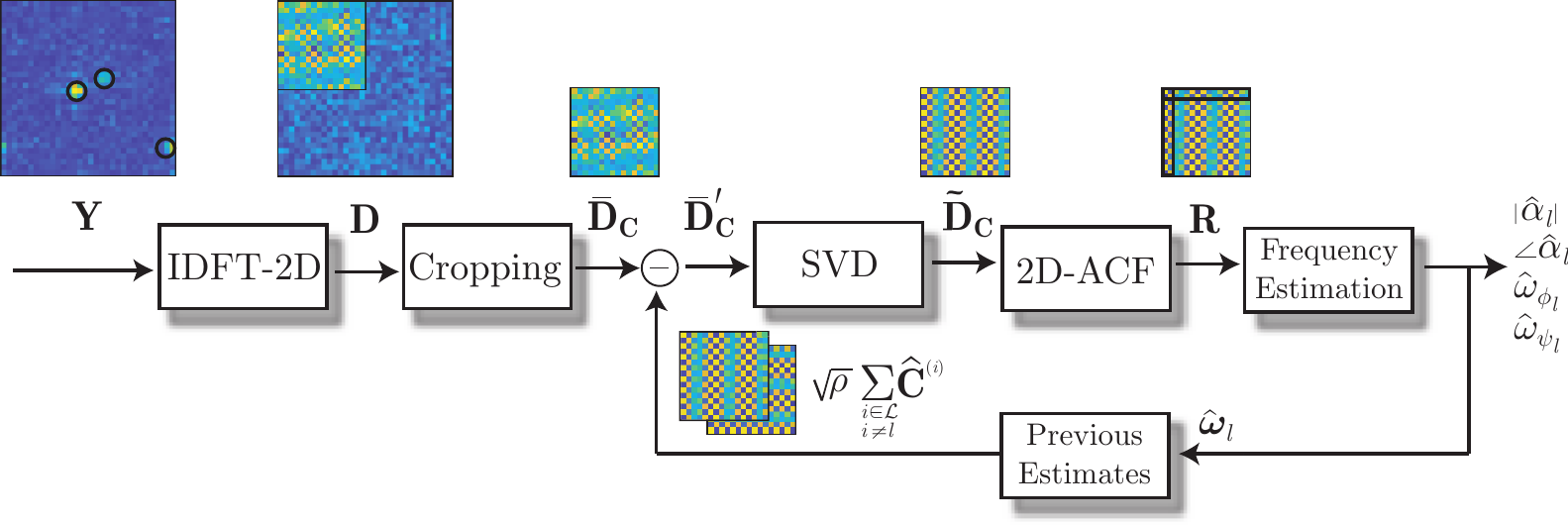}
\caption{{Block diagram of the proposed TSDCE algorithm.}}
\label{fig:diagram} 
\end{center}
\end{figure*}
\subsection{Path Extraction via Rank-one Approximations}
\label{sec:svd}

As previously discussed, since the informative part of the signal concentrates only on the top-left sub-matrix of $\mathbf{D}$, the first step is to crop such matrix to extract $\mathbf{\bar{D}}_{\mathbf{C}}$. Let us continue by obtaining an estimate of the most powerful path component $\mathbf{\bar{C}}^{(1)}$ via the rank-one approximation of $\mathbf{\bar{D}}_{\mathbf{C}}$ according to the dominant singular value
\begin{equation}
      \mathbf{\widetilde{D}}_\mathbf{C} = \check{s}_1 \check{\mathbf{u}}_1 \check{\mathbf{v}}_1^{H} = \sqrt{\rho}\mathbf{\bar{C}}^{(1)} + \mathbf{E} =  \sqrt{\rho}\hat{\mathbf{C}}^{(1)},
      \label{eq:Sform}
\end{equation}
where $\check{s}_1$, $\check{\mathbf{u}}_1$ and $\check{\mathbf{v}}_1$ are the first singular value and vectors extracted from the \ac{SVD} of $\mathbf{\bar{D}}_{\mathbf{C}}$ and $\mathbf{E}~=~ \sqrt{\rho} \left( \hat{\mathbf{C}}^{(1)}  - \mathbf{\bar{C}}^{(1)} \right)$ contains the residual error, which is later analyzed in Sect.~\ref{sec:bounds}.

\subsection{Sample Autocorrelation Function}
Let us consider the use of the unbiased 2D sample \ac{ACF} of the rank-one observation $\mathbf{\widetilde{D}}_{\mathbf{C}}$
\begin{equation}
\mathbf{R} = \mathrm{ACF}_{\mathrm{2D}}(\mathbf{\widetilde{D}}_{\mathbf{C}}) \in \mathbb{C}^{n_r \times n_t}
\end{equation}
with elements $r_{m,n}= \left[\mathbf{R}\right]_{m,n}$ given by
\begin{eqnarray}
    r_{m,n} &=& 
   \frac{1}{\kappa_{m,n}} \sum_{\mu=0}^{n_r-m-1}\sum_{\nu=0}^{n_t-n-1}\tilde{d}^{*}_{\mu,\nu}\tilde{d}_{\mu+m,\nu+n},
    \label{eq:autocorrelation0} \\
    \kappa_{m,n} &=& (n_r-m)(n_t-n),
\end{eqnarray}
where $\tilde{d}_{m,n} =\left[\mathbf{\widetilde{D}}_{\mathbf{C}}\right]_{m,n}$. Note that the 2D-\ac{ACF} can be efficiently computed using a zero-padded \ac{FFT} \cite{Smi07} and that, due to its symmetry, only positive lags are considered.
Using Eq.~(\ref{eq:Sform}) we can express $\mathbf{R}$ as a sum of auto and cross-correlation terms
\begin{equation}
\begin{split}
\mathbf{R} & =   \rho \mathrm{ACF}_{\mathrm{2D}}\left( \mathbf{\bar{C}}^{(1)}  \right) +   \mathrm{ACF}_{\mathrm{2D}}\left(\mathbf{E} \right) \\
& + \sqrt{\rho}\mathrm{CCF}_{\mathrm{2D}}\left( \mathbf{\bar{C}}^{(1)}, \mathbf{E} \right) + \sqrt{\rho}\mathrm{CCF}_{\mathrm{2D}}\left( \mathbf{E}, \mathbf{\bar{C}}^{(1)} \right),
\end{split}\nonumber
\end{equation}
where $\mathrm{CCF}_{\mathrm{2D}} \left(\mathbf{A},\mathbf{B}\right)$ stands for the sample 2D cross-correlation sequence between two matrices. 

By substituting $\tilde{d}_{\mu,\nu}$ in Eq.~\eqref{eq:autocorrelation0} by the sinusoidal component of Eq.~\eqref{eq:dbarnoise}, the corresponding \ac{ACF} results in the same sinusoid multiplied by the conjugate of its amplitude and zero initial phase. Thus, the above equation can be compacted into 

\begin{equation}
    \mathbf{R} = \rho A_1 ^{*}\mathbf{\bar{C}}^{(1)} +  \bm{\Gamma},
\end{equation}
where $ \bm{\Gamma}$ now condenses all the terms involving the residual $\mathbf{E}$. In general, the elements of $\mathbf{R}$ can be expressed as
\begin{equation}
    r_{m.n} = \rho |A_1|^2 e^{j(\omega_{\psi_1}m + \omega_{\phi_1}n)} + \gamma_{m,n},
    \label{eq:autocorrelem}
\end{equation}
$\gamma_{m,n} = [\bm{\Gamma}]_{m,n}$ is the residual error component at lag $(m,n)$.

The use of the \ac{ACF} introduces three main advantages. First, the effect of the phase of $\alpha_1$ is removed from the observation, facilitating the estimation of the underlying spatial frequencies. Second, the ratio between the power of the target sinusoidal signal and the residual is considerably reduced due to its averaging property, concentrating part of the noise at the $(0,0)$ lag where the phase is known to be zero. Finally, while the power of the residual $\mathbf{E}$ is evenly distributed across $\mathbf{\widetilde{D}}_\mathbf{C}$, the new error $\bm{\Gamma}$ is lower at small correlation lags. The power of the \ac{ACF} error increases uniformly with the lag due to the smaller number of samples taking place in the computation of the autocorrelation estimate. In fact, it holds for the unbiased estimator that $\mathrm{var}(r_{m,n}) \propto\left(1/\kappa_{m,n}\right)$ \cite{JenkinsBook68}. This will motivate the use of a \ac{WLS} approach for frequency estimation.

In any case, it can be shown that all the required information to estimate $\omega_{\psi_1}$ and $\omega_{\phi_1}$ is contained in the phase angle of $r_{m,n}$ \cite{Tre85}. In fact, since the phase of $\mathbf{R}$ is zero at $m=n=0$, the following estimates for the vectors in Eqs.~\eqref{eq:cbvectors1t}-\eqref{eq:cbvectors2t} can be obtained
\begin{eqnarray}
    \hat{\mathbf{c}}(\omega_{\psi_1}) &=& \frac{1}{\sqrt{n_r}} \left[e^{\angle r_{0,0}},\, e^{j\angle r_{1,0}},\dots, e^{j\angle r_{n_r-1,0}}\right]^{T}, \label{eq:cvectorest1}\\
    \hat{\mathbf{c}}(\omega_{\phi_1}) &=& \frac{1}{\sqrt{n_t}}\left[e^{\angle r_{0,0}},\, e^{j\angle r_{0,1}},\dots, e^{j\angle r_{0,n_t-1}}\right]^{T},
    \label{eq:cvectorest2}
\end{eqnarray}
evidencing that all the required frequency information corresponding to the extracted path component is embedded within the phase of the first row and first column of $\mathbf{R}$.

\subsection{Spatial Frequency Estimation}
By looking at the form of $\mathbf{\bar{c}}(\omega_{\psi_1})$ and $\mathbf{\bar{c}}(\omega_{\phi_1})$ in Eqs.~(\ref{eq:cbvectors1t}) and~(\ref{eq:cbvectors2t}), it is evident that the unwrapped phase on the vertical and horizontal directions is given, respectively, by
\begin{eqnarray}
 \mathcal{P}_{\psi_1}(m) &=& \omega_{\psi_1}m, \quad m = 0,\dots,n_r - 1, \\
 \mathcal{P}_{\phi_1}(n) &=& \omega_{\phi_1}n, \quad\,\, n = 0,\dots,n_t - 1.
\end{eqnarray}
Thus, the phases of the elements making up $\mathbf{\hat{c}}(\psi_1)$ and $\mathbf{\hat{c}}(\phi_1)$ correspond to wrapped images of $\mathcal{P}_{\psi_1}$ and $\mathcal{P}_{\phi_1}$, respectively. As a result, the frequency estimation problem boils down to estimating the slopes of the unwrapped phase sequences on both the vertical and horizontal directions. To this end, a simple phase unwrapping procedure followed by \ac{WLS} slope estimation is proposed. The frequency estimation steps are as follows.

\subsubsection{Phase unwrapping}
Initially, first-order phase differences are stored
\begin{eqnarray}
    \bm{\Delta}^{(\psi_1)} &\triangleq& [\Delta^{(\psi_1)}_0,\, \Delta^{(\psi_1)}_1, \dots, \Delta^{(\psi_1)}_{n_r-1}]^{T},\\
    \bm{\Delta}^{(\phi_1)} &\triangleq& [\Delta^{(\phi_1)}_0,\, \Delta^{(\phi_1)}_1, \dots, \Delta^{(\phi_1)}_{n_t-1}]^{T},
\end{eqnarray}
where $\Delta^{(\psi_1)}_0 = \Delta^{(\phi_1)}_0 = 0$, $\Delta^{(\psi_1)}_m~=~ \angle(r_{m,0}\cdot r^{*}_{m-1,0})$ and $\Delta^{(\phi_1)}_n~=~ \angle(r_{0,n}\cdot r^{*}_{0,n-1})$.

Note that, ideally, the phase differences between consecutive elements correspond to the angular frequencies $\omega_{\psi_1}$ and $\omega_{\phi_1}$ and, therefore, the phase steps are known to have a magnitude smaller than or equal to $\pi$, i.e. $|\Delta_{m}|\leq \pi$. In fact, due to noise effects, phase differences having a magnitude close to 0 or $\pi$ may result in unwanted wrappings. To mitigate this effect, the observed phase differences are wrapped to the range $[-\pi,\,\pi]$ or to $[0,\,2\pi]$, whichever provides the smaller variance. By taking the cumulative sum of $\bm{\Delta}^{(\psi_1)}$ and $\bm{\Delta}^{(\phi_1)}$, the unwrapped phases are estimated as
\begin{eqnarray}
    \hat{\mathcal{P}}_{\psi_1} (m) &=& \sum_{i=0}^{m}\Delta^{(\psi_1)}_i, \quad m = 0,\,1,\dots,n_r-1, \\
        \hat{\mathcal{P}}_{\phi_1} (n) &=& \sum_{i=0}^{n}\Delta^{(\phi_1)}_i, \quad n = 0,\,1,\dots,n_t-1.
\end{eqnarray}

\subsubsection{\ac{WLS}}
We formulate the \ac{WLS} optimization problem
\begin{eqnarray} 
    \hat{\omega}_{\psi_1} &=& \arg\min_{\omega_{\psi_1}} \left\{ \sum_{m=0}^{n_r-1}w_m \left(\hat{\mathcal{P}}_{\psi_1} (m) - \omega_{\psi_1}m \right)^{2} \right\},\\
    \hat{\omega}_{\phi_1} &=& \arg\min_{\omega_{\phi_1}} \left\{ \sum_{n=0}^{n_t-1}w_n \left(\hat{\mathcal{P}}_{\phi_1} (n)- \omega_{\phi_1}n \right)^{2}\right\},
\end{eqnarray}
where $w_m$ and $w_n$ are the selected weights. The solution for both frequencies has the common form
\begin{equation}
        \hat{\omega} = \frac{\sum_{i=0}^{M-1}w_i(i -  \bar{\mathrm{x}})( \hat{\mathcal{P}}(i)-\bar{\mathrm{y}})}{\sum_{i=0}^{M-1}w_i (i -  \bar{\mathrm{x}})^2},
    \label{eq:lsfreq}
\end{equation}
with
\begin{equation}
   \bar{\mathrm{x}} = \frac{\sum_{i=0}^{M-1}w_i i}{\sum_{i=0}^{M-1}w_i}, \quad
    \bar{\mathrm{y}} = \frac{\sum_{i=0}^{M-1}w_i \hat{\mathcal{P}}(i)}{\sum_{i=0}^{M-1}w_i}, \label{eq:lsxy}
\end{equation}
and $M \in \{n_t,\, n_r\}$, $i \in \{m,\,n\}$ and $\hat{\mathcal{P}}(i)$ may refer to $\hat{\mathcal{P}}_{\psi_1}(m)$ or $\hat{\mathcal{P}}_{\phi_1}(n)$, correspondingly, so that $\hat{\omega}$ is an estimate of $\hat{\omega}_{\psi_1}$ or  $\hat{\omega}_{\phi_1}$. It is important to note that, 
according to Eq.~(\ref{eq:omegaphil}), the estimated frequencies must be wrapped to the range $[-\pi,\, \pi]$. Selecting the proper weights is discussed next.

\subsubsection{Weights}

Taking into account the observations from the autocorrelation values $r_{m,n}$, with $\mathrm{var}(r_{m,n}) \propto \left(1/\kappa_{m,n}\right)$ and that, under relatively high \ac{SNR}, the phase noise variance of a noisy complex sinusoid can be assumed to be proportional to the noise variance \cite{Fu08}, we can derive proper weights for the \ac{WLS} problem as follows. Assuming that the variance of the sample \ac{ACF} estimator is independent for different lags, the variance of first-order differences is proportional to the product of the variances corresponding to neighboring autocorrelation values in $\Delta_{m}^{(\psi_1)}$ and  $\Delta_{n}^{(\phi_1)}$
\begin{equation}
    \mathrm{var}(\Delta_{i}) \propto \left( \frac{1}{(M-i)(M-i+1)} \right).
\end{equation}
As phase differences are accumulated during the unwrapping operation, the variances of $\hat{\mathcal{P}}_{\psi_1}$ and $\hat{\mathcal{P}}_{\phi_1}$ should approximate
\begin{equation}
\begin{split}
    \mathrm{var}\left(\hat{\mathcal{P}}(i)\right) &\propto \left( \sum_{k=0}^{i} \frac{1}{(M-k)(M-k+1)} \right) \\& = \left( \frac{i+1}{(M+1)(M-i)} \right).
    \end{split}
\end{equation}

Since optimum weights should be proportional to the inverse of the variances, the weights to be used are given by
\begin{equation}
    w_i = \frac{(M+1)(M-i)}{i+1}.
    \label{eq:weights}
\end{equation}

\subsection{Estimation of the Path Complex Coefficient}
\label{sec:alphaest}
First, taking into account Eq.~(\ref{eq:autocorrelem}), an estimate of $|A_1|^2$ can be obtained by calculating the weighted average autocorrelation magnitude at lags different from $m=n=0$. As in the frequency estimation step, the weights are selected according to the relative variance of the autocorrelation estimates
\begin{eqnarray}
    |\hat{A}_1|^2 &=& \frac{1}{\rho}\mathcal{K} \mathop{\sum^{n_r-1}\sum^{n_t-1}}_{m,n \neq (0,0)}\kappa_{m,n}|r_{m,n}|,
    \label{eq:Aest}
    \\
  \mathcal{K} &=& \frac{1}{ \frac{1}{4}n_r(n_r+1)n_t(n_t+1)-n_t n_r},
\end{eqnarray}
where $\mathcal{K} = \mathop{\sum \sum}_{m,n \neq (0,0)}\kappa_{m,n}$ is the weight normalization constant. Then, according to Eq.~(\ref{eq:autocorrelem}), the magnitude of the path coefficient can be obtained as
\begin{equation}
    |\hat{\alpha}_1| = \sqrt{n_t n_r}\sqrt{|\hat{A}_1|^2}.
\end{equation}

Finally, the maximum likelihood estimate of the phase of $\alpha_1$ is obtained by taking the mean of the underlying circular normal distribution
\begin{equation}
    \angle \hat{\alpha}_1 = \angle \left( \frac{1}{n_r n_t}\sum_{m=0}^{n_r-1}\sum_{n=0}^{n_t-1}\tilde{d}_{m,n}e^{-j(\hat{\omega}_{\psi_1}m + \hat{\omega}_{\phi_1}n)} \right).
    \label{eq:angleest}
\end{equation}

\subsection{Successive Path Estimation}
\label{sec:sic}
Previous sections presented the different processing stages aimed at estimating the parameters corresponding to the dominant path within the original observation. The estimation of the rest of paths is achieved by following the same processing steps but from a modified observation, where the previous estimated components have been suppressed following a \ac{SIC} approach. Let us assume that one or several paths have already been estimated, leading to estimated parameters $\bm{\hat{\omega}}_i$, $i\in\mathcal{L}$, where $\mathcal{L}$ is the set of previously estimated paths. By using such parameters, each estimated component $\hat{\mathbf{C}}^{(i)}$ can be reconstructed as
\begin{equation}
    [\hat{\mathbf{C}}^{(i)}]_{m,n} = \frac{|\hat{\alpha}_i|}{\sqrt{n_t n_r}}e^{j\angle\hat{\alpha}_i} e^{j(\hat{\omega}_{\psi_i}m + \hat{\omega}_{\phi_i}n)}.
\end{equation}

\begin{algorithm}[t]
\SetAlgoLined
\SetKwData{Left}{left}\SetKwData{This}{this}\SetKwData{Up}{up}
\SetKwFunction{Union}{Union}\SetKwFunction{FindCompress}{FindCompress}
\SetKwInOut{Input}{Input}\SetKwInOut{Output}{Output}
\SetKwInOut{Initialize}{Initialize}

\Input{$\mathbf{Y}$, $L_d$, $K$, $n_t$, $n_r$, $\rho$}
\Output{$\bm{\hat{\omega}}$ }
\Initialize{$\mathcal{L}=\emptyset$; }
$\mathbf{D} = \mathrm{IDFT}_{\mathrm{2D}}\left\lbrace \mathbf{Y} \right\rbrace$\;
$\mathbf{\bar{D}}_{\mathbf{C}} = \mathbf{D}_{0:n_r-1,0:n_t-1}$\; 

\For{$k=1$ to $K$}{
\For{$l=1$ to $L_d$}{
    $  \mathbf{\bar{D}}^{'}_{\mathbf{C}} \leftarrow \mathbf{\bar{D}}_{\mathbf{C}} - \sqrt{\rho}\sum_{\substack{i\in\mathcal{L} \\ i\neq l}} \hat{\mathbf{C}}^{(i)}$
   \eIf{$k=1$ $\&$ $l<L_d$}{
    $[\mathbf{U},\mathbf{S},\mathbf{V}] = \mathrm{SVD}(\mathbf{\bar{D}}^{'}_{\mathbf{C}})$\;
  $\mathbf{\widetilde{D}}_{\mathbf{C}} = \check{s}_1 \check{\mathbf{u}}_1
  \check{\mathbf{v}}_1^{H}$\;}
  {
  $\mathbf{\widetilde{D}}_{\mathbf{C}} =\mathbf{\bar{D}}^{'}_{\mathbf{C}}$\;
  }
   $\mathbf{R}  = \mathrm{ACF}_{\mathrm{2D}}(\mathbf{\widetilde{D}}_{\mathbf{C}})$\;
  $\Delta^{(\psi_l)}_m~=~ \angle(r_{m,0}\cdot r^{*}_{m-1,0}), \quad m=1,\dots,n_r-1$\;
 $\Delta^{(\phi_l)}_n~=~ \angle(r_{0,n}\cdot r^{*}_{0,n-1}), \quad n=1,\dots,n_t-1$\;
   \If{$\mathrm{var}(\bm{\Delta}^{(\psi_l)}) > \mathrm{var}(\mathcal{W}_{[0,2\pi]}(\bm{\Delta}^{(\psi_l)}))$}{
    $\bm{\Delta}^{(\psi_l)} \leftarrow \mathcal{W}_{[0,2\pi]}(\bm{\Delta}^{(\psi_l)})$
    }
   \If{$\mathrm{var}(\bm{\Delta}^{(\phi_l)}) > \mathrm{var}(\mathcal{W}_{[0,2\pi]}(\bm{\Delta}^{(\phi_l)}))$}{
    $\bm{\Delta}^{(\phi_l)} \leftarrow \mathcal{W}_{[0,2\pi]}(\bm{\Delta}^{(\phi_l)})$ 
    }
  $\hat{\mathcal{P}}_{\psi_l} (m) = \sum_{i=0}^{m}\Delta^{(\psi_l)}_i, \quad m = 0,\,1,\dots,n_r-1$\;
 $\hat{\mathcal{P}}_{\phi_l} (n) = \sum_{i=0}^{n}\Delta^{(\phi_l)}_i, \quad n = 0,\,1,\dots,n_t-1$\;
  Estimate $\hat{\omega}_{\psi_l}$ and $\hat{\omega}_{\phi_l}$ with Eqs.~(\ref{eq:lsfreq}),\eqref{eq:lsxy},\eqref{eq:weights}\;
  $\hat{\omega}_{\phi_l} = \mathcal{W}_{[-\pi,\pi]}(\hat{\omega}_{\phi_l})$ ; $\hat{\omega}_{\psi_l} = \mathcal{W}_{[-\pi,\pi]}(\hat{\omega}_{\psi_l})$\;
  Estimate $\hat{\alpha}_l$ with Eqs.~(\ref{eq:Aest})-(\ref{eq:angleest})\;
 $\mathcal{L} \leftarrow \mathcal{L}\cup\{l\}$\;
     }  
 }
 \caption{TSDCE Algorithm} \label{Alg:Alg1}
\end{algorithm}

By suppressing the above reconstructed path components from the original observation, the spatial domain observation matrix to estimate path component $l$ can be updated as follows
\begin{equation}
    \mathbf{\bar{D}}^{'}_{\mathbf{C}} \leftarrow \mathbf{\bar{D}}_{\mathbf{C}} - \sqrt{\rho}\sum\limits_{\substack{i\in\mathcal{L} \\ i\neq l}} \hat{\mathbf{C}}^{(i)}, \label{eq:Scupdate}
\end{equation}
where the updated observation $\mathbf{\bar{D}}^{'}_{\mathbf{C}}$ will then become the new input replacing $\mathbf{\bar{D}}_{\mathbf{C}}$ in the processing stages described throughout Sections~\ref{sec:svd} to~\ref{sec:alphaest}. 

The steps of the proposed method are summarized in Algorithm 1. It is worth noting that, while $L$ iterations are sufficient for having initial estimates of the parameters of the $L$ desired paths, these can be further refined by running additional estimation rounds through an outer loop with $K>1$. Indeed, once the estimates of each individual path component are available after the first estimation round ($k=1$), these can be effectively used to cancel all the disturbing path contributions from the original observation, leading to better estimates than the initial ones. Such cancellation is more effective than the SVD extraction performed in the first round, where all the parameters need to be estimated from scratch without any a priori path information. In this context, the exact value of $L$ is often assumed to be known \cite{Mon15} or, alternatively, a desired number of path components to be extracted $L_d$ is set \cite{Alk14}, as reflected in Algorithm 1. Note, however, that $L_d$ is used as a stopping criterion in TSDCE without affecting its applicability in real scenarios, where a power-based criterion on the extracted components could be used.

\section{Analysis of Performance Bounds}
\label{sec:bounds}
This section derives the upper and lower performance bounds of the proposed TSDCE algorithm by using the mean SSE as a metric. In the upper bound case, separate lemmas are given for $L=1$ and $L > 1$ cases.
\subsection{Upper Bound Analysis}
\begin{lemma} \label{lem:l1ub}
For the case $L=1$, an upper bound for the mean \ac{SSE} of the proposed method is given by $\frac{(\sqrt{n_r} + \sqrt{n_t})^2}{\mathrm{SNR}_{\mathbf{C}}}$.
\end{lemma}
\begin{proof}
The proposed method departs from the cropped spatial domain observation $\mathbf{\bar{D}}_\mathbf{C}$, which has been shown in Lemma \ref{lem:channel} to provide a noisy observation of the channel matrix. An upper bound for the proposed method can be established by studying the mean \ac{SSE} corresponding to the channel estimate derived from the SVD-based rank-one approximation $\mathbf{\widetilde{D}}_\mathbf{C}$. 

Let us first analyze the variance of the elements $\tilde{d}_{m,n}$ of $\mathbf{\widetilde{D}}_\mathbf{C}$. Recalling that the approximation is based on the largest singular value, the variance is given by
\begin{equation}
    \mathrm{var}(\tilde{d}_{m,n}) = \frac{\mathbb{E}\left\{\left\|\mathbf{\widetilde{D}}_\mathbf{C} \right\|_{F}^{2}\right\} }{n_t n_r} = \frac{\mathbb{E}\left\{\lambda_1(\mathbf{\bar{D}}_\mathbf{C}^{H}\mathbf{\bar{D}}_\mathbf{C})\right\}}{n_t n_r}.
    \label{eq:vartildesc}
\end{equation}
Taking into account the expectation operator and that path components and noise are uncorrelated, the impact of signal and noise cross-terms can be neglected. Therefore, by considering Weyl's inequality for the eigenvalues of Hermitian matrices \cite{knutson2001}, it follows that
\begin{equation}
       \mathbb{E}\left\{\lambda_1(\mathbf{\bar{D}}_\mathbf{C}^{H}\mathbf{\bar{D}}_\mathbf{C}) \right\}\leq \mathbb{E}\left\{\lambda_1 (\rho\mathbf{\bar{C}}^{H}\mathbf{\bar{C}})\right\} + \mathbb{E}\left\{\lambda_1 (\mathbf{\bar{Z}}^{H}\mathbf{\bar{Z}})\right\},\label{eq:Wey}
\end{equation}
indicating an upper bound for the largest singular value (recall that the largest singular value is obtained as the square root of $\lambda_1(\mathbf{\bar{D}}_\mathbf{C}^{H}\mathbf{\bar{D}}_\mathbf{C})$) of the noisy observation as a function of the largest eigenvalue of the path component matrix product and of the noise matrix product. For matrix $\rho\mathbf{\bar{C}}^{H}\mathbf{\bar{C}}$ and assuming $L=1$, it holds
\begin{equation}
    \mathbb{E}\left\{\lambda_1 (\rho\mathbf{\bar{C}}^{H}\mathbf{\bar{C}})\right\}=\mathbb{E}\left\{\rho|\alpha_1|^2\right\}=\rho.
    \label{eq:eigC}
\end{equation}

On the other hand, the variance of the noise can be expressed as
\begin{equation}
    \sigma^2_z = \frac{1}{n_t n_r}\mathbb{E}\left\{\left\|\mathbf{\bar{Z}} \right\|_{F}^{2}\right\} = \frac{1}{n_t n_r}\sum_{i=1}^{n_r}\mathbb{E}\left\{\lambda_i(\mathbf{\bar{Z}}^{H}\mathbf{\bar{Z}})\right\}.
\end{equation}
According to Gordon's theorem for random Gaussian matrices \cite{Gordon85}, the mean of the largest eigenvalue of $\mathbf{\bar{Z}}^{H}\mathbf{\bar{Z}}$, i.e. 
$\lambda_1(\mathbf{\bar{Z}}^{H}\mathbf{\bar{Z}})$, is bounded by
\begin{equation}
    \mathbb{E}\left\{\lambda_1(\mathbf{\bar{Z}}^{H}\mathbf{\bar{Z}})\right\} \leq \sigma^2_z(\sqrt{n_r} + \sqrt{n_t})^2.
   \end{equation}

By using the above result in Eq.~(\ref{eq:Wey}) and substituting in Eq.~(\ref{eq:vartildesc})
\begin{equation}
     \mathrm{var}(\tilde{d}_{m,n}) \leq \frac{\left(\rho + \sigma^2_z(\sqrt{n_t} + \sqrt{n_r})^2\right)}{n_t n_r}.
\end{equation}

Then, the variance added by the residual noise $\mathbf{E}$ in Eq.~(\ref{eq:Sform}) must satisfy 
\begin{equation}
    \sigma^2_{e} \leq \frac{(\sqrt{n_r} + \sqrt{n_t})^2}{n_t n_r}\sigma^2_z. 
\end{equation}
Therefore, as in Eq.~(\ref{eq:Hdest}), an improved estimate of the channel due to the SVD noise filtering effect can be obtained as $\mathbf{\hat{H}}_{\mathbf{\widetilde{D}}}~=~\sqrt{\frac{n_t n_r}{\rho}}\mathbf{\widetilde{D}}_{\mathbf{C}}$, leading to a mean \ac{SSE}
\begin{equation}\label{eq:sse_s2}
\begin{split}
   &\mathbb{E} \left\{ \left\|\mathbf{\hat{H}}_{\mathbf{\widetilde{D}}}-\mathbf{H}\right \|_F^2 \right\}<  \mathbb{E} \left\{ \sum_{m=0}^{n_r-1}\sum_{n=0}^{n_t-1} \left(\hat{h}_{m,n}-h_{m,n}\right)^2 \right\}=
   \nonumber \\&= \sum_{m=0}^{n_r-1}\sum_{m=0}^{n_t-1}\mathbb{E} \left\{\check{e}_{m,n}^2 \right\}= n_t n_r  \sigma^2_{\check{e}},
   \end{split}
    \end{equation}
where $\check{e}_{m,n} = \sqrt{\frac{n_t n_r}{\rho}}[\mathbf{E}]_{m,n}$ and $\sigma^2_{\check{e}}=\mathrm{var}(\check{e}_{m,n})=\frac{n_t n_r}{\rho}\sigma^2_e$. Note that the above result can alternatively be expressed as
\begin{equation}
\begin{split}
 &\mathbb{E} \left\{ \left\|\mathbf{\hat{H}}_{\mathbf{\widetilde{D}}}-\mathbf{H}\right \|_F^2 \right\} < n_t n_r \frac{n_t n_r}{\rho}\sigma^2_e =
\\& =\frac{n_t n_r}{\rho}(\sqrt{n_r} + \sqrt{n_t})^2\sigma^2_z = \frac{(\sqrt{n_r} + \sqrt{n_t})^2}{\mathrm{SNR}_{\mathbf{C}}}.
\end{split}
\end{equation}
\end{proof}

\begin{lemma} \label{lem:lh1ub}
For $L > 1$, an upper bound for the mean \ac{SSE} is given by $\frac{n_t n_r}{\rho}\sum_{l=1}^{L}n_t \lambda_{l,\mathrm{mean}}$, where $\lambda_{l,\mathrm{mean}}$ corresponds to the mean of the ordered statistics of the normalized eigenvalue distribution of the noise, i.e.  $\lambda_{l,\mathrm{mean}} = \mathbb{E}\left\{\lambda_l\left(\frac{1}{n_t}\mathbf{\bar{Z}}^{H}\mathbf{\bar{Z}}\right) \right\}$.
\end{lemma}

\begin{proof}
Similarly to Lemma \ref{lem:l1ub}, we consider in this case the $L$-rank approximation of $\mathbf{\bar{D}}_{\mathbf{C}}$, denoted here as $\mathbf{\widetilde{D}}_{\mathbf{C}_L}$. In this case
\begin{equation}
    \mathrm{var}(\tilde{d}_{{m,n}^{(L)}}) = \frac{\mathbb{E}\left\{\left\|\mathbf{\widetilde{D}}_{\mathbf{C}_L} \right\|_{F}^{2}\right\} }{n_t n_r} = \frac{\mathbb{E}\left\{\sum_{l=1}^{L}\lambda_l(\mathbf{\bar{D}}_\mathbf{C}^{H}\mathbf{\bar{D}}_\mathbf{C})\right\}}{n_t n_r},
    \label{eq:vartildesc2}
\end{equation}
where $\tilde{d}_{{m,n}^{(L)}}=[\mathbf{\widetilde{D}}_{\mathbf{C}_L}]_{m,n}$. By applying the Ky Fan inequality \cite{moslehian2012ky} for Hermitian matrices
\begin{equation}
    \sum_{l=1}^{L}\mathbb{E}\left\{\lambda_l(\mathbf{\bar{D}}_\mathbf{C}^{H}\mathbf{\bar{D}}_\mathbf{C})\right\}  \leq \sum_{l=1}^{L} \mathbb{E}\left\{\lambda_l (\rho\mathbf{\bar{C}}^{H}\mathbf{\bar{C}})\right\} + 
     \sum_{l=1}^{L}\mathbb{E}\left\{\lambda_l (\mathbf{\bar{Z}}^{H}\mathbf{\bar{Z}})\right\}.
\end{equation}
As in Eq.~(\ref{eq:eigC}), the first term $\sum_{l=1}^{L} \mathbb{E}\left\{\lambda_l (\rho\mathbf{\bar{C}}^{H}\mathbf{\bar{C}})\right\}=\rho$, while for the second term it becomes convenient to study the distribution of the eigenvalues of the normalized Gaussian random matrix $\frac{1}{n_t}\mathbf{\bar{Z}}^{H}\mathbf{\bar{Z}}$, which asymptotically follows the Marchenko-Pastur density \cite{pastur1967}. Without loss of generality, considering $n_t \geq n_r$, the eigenvalue \ac{PDF} can be written as
\begin{equation}
    f_{\lambda}(x) = \frac{1}{2\pi\sigma_z^2}\frac{\sqrt{(x-a)(b-x)}}{cx},
\end{equation}
where $c = n_r/n_t$ and $a,b = \sigma_z^2(1 \pm \sqrt{c})^2$. The \ac{CDF} is then given by
\begin{equation}
    F_{\lambda}(x) = \int_{a}^{x}  f_{\lambda}(x) dx = \dot{F}(x) - \dot{F}(a) 
\end{equation}
where
\begin{eqnarray}
    \dot{F}(x) &=& \sqrt{(x-a)(b-x)} +\frac{a+b}{2}\arcsin{\left(\frac{2x-a-b}{b-a}\right)} \nonumber \\
    &-& \frac{ab}{\sqrt{ab}}\arcsin{\left(\frac{(a+b)x -2ab}{x(b-a)}\right)}. 
\end{eqnarray}
Note, however, that the algorithm relies on the magnitude-ordered singular values, so we are interested in the distributions of the order statistics for the $n_r$ eigenvalues of $\frac{1}{n_t}\mathbf{\bar{Z}}^{H}\mathbf{\bar{Z}}$, with \ac{PDF}
\begin{equation}
\begin{split}
    f_{\lambda_l}(x) & = (n_r-l+1)\binom{n_r}{n_r-l+1} f_{\lambda}(x)\\
    &  \left(F_{\lambda}(x) \right)^{n_r-l}\left(1-F_{\lambda}(x) \right)^{l-1},  
\end{split}
\end{equation}
where the subindex $\lambda_l$ indicates that such PDF corresponds to the $l$-th largest eigenvalue from $n_r$ samples of the underlying eigenvalue distribution. The above formula allows to compute numerically the mean for each of the magnitude-descending-ordered eigenvalues $\lambda_{l,\mathrm{mean}} = \mathbb{E}\left\{\lambda_l\left(\frac{1}{n_t}\mathbf{\bar{Z}}^{H}\mathbf{\bar{Z}}\right) \right\}$, so that
\begin{equation}
      \mathrm{var}(\tilde{d}_{{m,n}^{(L)}}) \leq \frac{\left(\rho + \sum_{l=1}^{L}n_t\lambda_{l,\mathrm{mean}}\right)}{n_t n_r},
\end{equation}
with a residual variance bounded by  $\sigma^2_{e} \leq \frac{\sum_{l=1}^{L}n_t \lambda_{l,\mathrm{mean}}}{n_t n_r}$. 

Let us write the channel estimate derived from $\mathbf{\widetilde{D}}_{\mathbf{C}_L}$ as $ \mathbf{\hat{H}}_{\mathbf{\widetilde{D}}_L} =  \sqrt{\frac{n_t n_r}{\rho}}\mathbf{\widetilde{D}}_{\mathbf{C}_L} = \mathbf{H} + \mathbf{E}_L$, with residual variance $\sigma^2_{\check{e}}=\frac{n_t n_r}{\rho}\sigma^2_e$, so that the mean \ac{SSE}, $\mathbb{E} \left\{ \left\|\mathbf{\hat{H}}_{\mathbf{\widetilde{D}}_L}-\mathbf{H}\right \|_F^2 \right\}$, is
\begin{equation}
\begin{split}
& \mathbb{E}\left\{\left\| \mathbf{E}_L \right\|^2_F\right\} < n_t n_r \frac{n_t n_r}{\rho} \sigma^2_{e}~=~\frac{n_t n_r}{\rho}\sum_{l=1}^{L}n_t \lambda_{l,\mathrm{mean}}.
\end{split}
\end{equation}

\end{proof}
\subsection{Lower Bound Analysis}
\begin{theorem}
A lower bound for the mean \ac{SSE} of the TSDCE algorithm is given by the estimator that, departing from the low-rank approximation of the spatial domain observation $\mathbf{\bar{D}}_{\mathbf{C}}$, achieves the \ac{CRLB} of the involved parameters $\bm{\omega}$.
\end{theorem}
\begin{proof}
The proposed TSDCE algorithm is based on the estimation of the parameters corresponding to the multiple two-dimensional frequencies constituting the channel. By considering the \ac{SVD}-denoised channel observation of Lemma \ref{lem:lh1ub}, a lower bound would be that corresponding to an estimator of the involved channel parameters achieving the \ac{CRLB}.

Consider the elements of the noisy channel observation
\begin{equation}
     \hat{h}_{m,n} = [\mathbf{\hat{H}}_{\mathbf{\widetilde{D}}_L}]_{m,n}=\sum_{l=1}^{L}|\alpha_l|e^{j\left(\angle \alpha_l + \omega_{\psi_l}m + \omega_{\phi_l}n \right)} + \check{e}_{m,n},
\end{equation}
where it has been shown in Lemma \ref{lem:lh1ub} that 
$\sigma^2_{\check{e}}=\frac{1}{\rho}\sum_{l=1}^{L}n_t \lambda_{l,\mathrm{mean}}$. Considering the $4L\times 1$ real vector of unknown parameters $\bm{\omega}$, the corresponding $4L \times 4L$ Fisher information matrix $\bm{\mathcal{F}}$ is 
\begin{equation}
    [\bm{\mathcal{F}}]_{i,j}=-\mathbb{E}\left\{\frac{\partial^2}{\partial\omega_i\partial\omega_j}\log\left(f\left(\mathrm{vec}(\mathbf{\hat{H}}_{\mathbf{\widetilde{D}}_L})|\bm{\omega}\right)\right) \right\},
\end{equation}
where $\omega_i$ is the $i$-th element of $\bm{\omega}$ and $f\left(\mathrm{vec}(\mathbf{\hat{H}}_{\mathbf{\widetilde{D}}_L})|\bm{\omega}\right)$ is the \ac{PDF} of $\mathrm{vec}(\mathbf{\hat{H}}_{\mathbf{\widetilde{D}}_L})$. 
A  zero-mean complex white Gaussian distribution for the elements of $\mathbf{E}_L$ is assumed for mathematical tractability, which has been empirically verified to hold for the considered channel model. In this case, the entries of the Fisher information matrix can be expressed as
\begin{equation}
     [\bm{\mathcal{F}}]_{i,j}=\frac{1}{\sigma^2_{\check{e}}}2\mathrm{Re} \left[\frac{\partial\mathrm{vec}(\mathbf{H})^{H}}{\partial\omega_i}\frac{\partial\mathrm{vec}(\mathbf{H})}{\partial\omega_j}   \right].
\end{equation}

The form of the derivatives involved in each element $[\bm{\mathcal{F}}]_{i,j}$ was derived in Appendix A of \cite{Hua92}. The \ac{CRLB} of the variance of the unbiased estimate for each parameter can therefore be obtained from the corresponding diagonal element of $\bm{\mathcal{F}}^{-1}$. 
\end{proof}

Since $\bm{\mathcal{F}}$ does not have a closed-form inverse, and the derivatives are dependent on the actual value of the parameters, the bound is obtained numerically through Monte-Carlo simulation. First, the CRLB variance is obtained for every parameter in each  channel realization. Then, the channel is reconstructed from the ground-truth channel parameters, each of them corrupted by noise with a variance corresponding to that of its respective CRLB.

\section{Experiments and Complexity Analysis}\label{Sec:Complexity}
In this section, we assess the performance of the TSDCE algorithm\footnote{Code available at https://github.com/SandraRoger/tsdce} through numerical experiments. The method is compared in terms of performance and complexity with the \ac{LS} and \ac{OMP} channel estimation methods in \cite{lee2014}, the DFT-based scheme (DFT-CEA) in \cite{Mon15}, and with the \ac{KF}-based beamtracking approach in \cite{Zha16}. To match an \ac{ABF} architecture, all the methods are evaluated with a single \ac{RF} chain. In the \ac{OMP} case, the number of grids is set to $G=180$ and the number of iterations is set to $I_t=20$ \cite{lee2014}. As in \cite{Mon15}, we set $N_{\text{DFT}}=1024$ points for the DFT-CEA method. Note that all methods are tested under the same pilot overhead, which is directly $Q\times P$ pilots. Also, the number of desired paths is set to $L_d=L$ for all methods. The simulations consider $1000$ random channel realizations for each \ac{SNR} value, following the observation model in Eq.~\eqref{eq:observation} with $\rho = 1$, i.e. the \ac{SNR} definition is $1/\sigma^2_n$. Channel coefficients are drawn from a zero mean complex Gaussian distribution with variance $1/L$, while \ac{AoA} and \ac{AoD} angles are drawn from a uniform distribution in the range $[0,\, \pi]$.

\begin{figure}[t]
\begin{center}
\includegraphics[width =0.7\columnwidth]{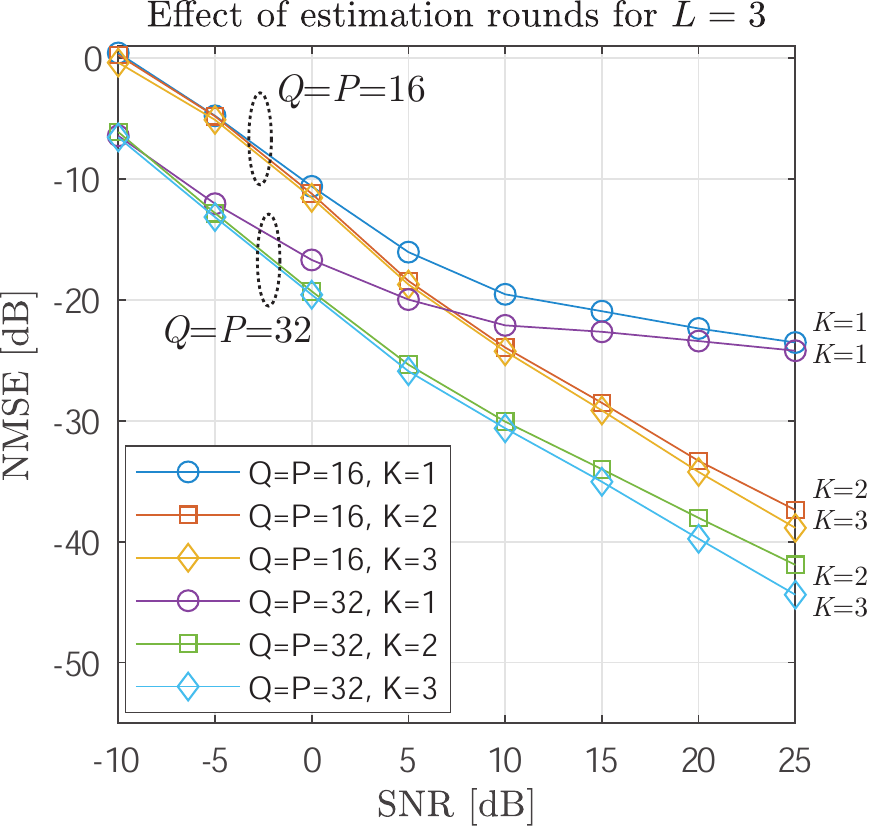}
\caption{{Performance comparison for different values of $K$ in the proposed TSDCE algorithm.}}\label{fig:results3} 
\end{center}
\end{figure}

We first assess the performance of the \ac{AoA}, \ac{AoD} and the path complex coefficients estimation jointly through the \ac{NMSE} of $\mathbf{H}$, defined as $\mathrm{NMSE}\,\textrm{[dB]}  =10\log_{10} \left(\mathbb{E} \left\{\frac{ \left\|\mathbf{\hat{H}}-\mathbf{H}\right \|_F^2} {\left\|\mathbf{H}\right \|_F^2} \right\}\right)$, where $\mathbf{\hat{H}}$ is the estimated channel and the expectation is approximated by averaging. Fig.~\ref{fig:results3} shows the effect of the $K$ parameter in the \ac{TSDCE}, which is directly related with the number of channel estimation rounds for refinement. The \ac{NMSE} is evaluated for different values of \ac{SNR} in a system with $n_t=n_r=16$ and considering codebooks of two different sizes, namely $Q=P=16$ and $Q=P=32$, and $L=3$. Recall that setting $K=1$ provides an initial estimation for the $L$ paths. By running either one or two additional estimation rounds ($K=2$ and $K=3$, respectively), the performance is notably enhanced due to the possibility of removing all the remaining paths from the observation in every iteration. When comparing the $K=2$ and $K=3$ cases, a higher $K$ improves the performance only slightly at high SNRs. Since choosing a higher $K$ increases also the complexity, we set $K=L_d$ in the next simulations for a good performance vs. complexity tradeoff.

Fig.~\ref{fig:results1} and Fig.~\ref{fig:results2} compare the different methods for $n_t=n_r=16$ and codebook with $Q=P=16$ and $Q=P=32$ elements, respectively, and three different values of $L$ (from 1 to 3). It can be observed that the proposed approach outperforms both the \ac{LS} and \ac{OMP} in all cases, and the DFT-CEA approach except at very low SNR values. Another important result is that the TSDCE performance achieves the \ac{CRLB} at medium to high \acp{SNR} in the multi-path cases. When comparing the results with different codebook sizes, an \ac{SNR} gain of approximately $6$~dB is observed, which matches the theoretical \ac{SNR} gain derived in Eq.~\eqref{eq:snrgain}.

\begin{figure*}[t!]
\begin{center}
\includegraphics[width =0.85\textwidth]{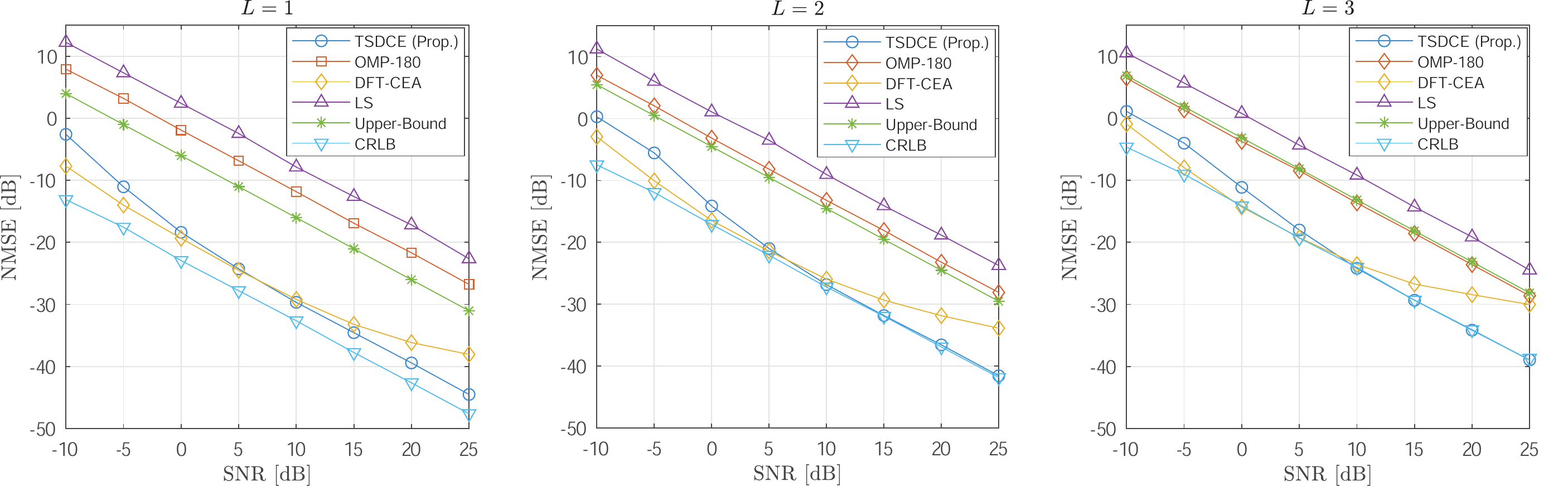}
\caption{{\ac{NMSE} of channel estimation for different values of \ac{SNR} with $n_t = n_r = 16$ and $Q=P=16$.}}\label{fig:results1} 
\end{center}
\end{figure*}
\begin{figure*}[t!]
\begin{center}
\includegraphics[width =0.85\textwidth]{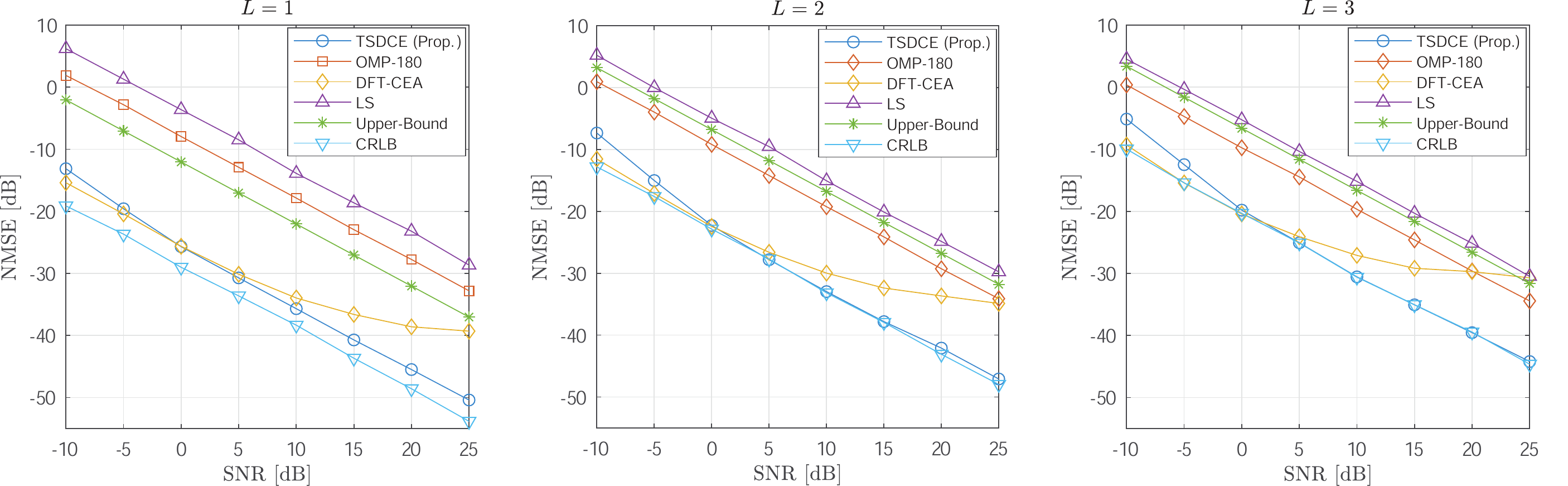}
\caption{{\ac{NMSE} of channel estimation for different values of \ac{SNR} with $n_t = n_r = 16$ and $Q=P=32$.}}\label{fig:results2} 
\end{center}
\end{figure*}

Fig.~\ref{fig:resultsDOA} analyzes the \ac{DoA} estimation capability of the TSDCE and DFT-CEA methods in the setups of Fig.~\ref{fig:results1} and Fig.~\ref{fig:results2} with $L=3$. To this end, the \ac{RMSE} of the angle measurements, calculated as $\sqrt{\frac{1}{|\mathcal{N}|}\sum_{i\in |\mathcal{N}|}(\vartheta_i-\hat{\vartheta}_i)^2 }$, and the probability of detection ($P_D$) at different \acp{SNR} are shown, considering as successful detections those with an \ac{RMSE}~$\leq 1^{\circ}$. According to such definition, the set of successfully detected angles is denoted by $\mathcal{N}$, while $\vartheta_i$ considers any AoA or AoD. At low \acp{SNR}, the DFT-CEA outperforms the TSDCE since it exhibits a higher $P_D$ than the TSDCE with similar \ac{RMSE}. At medium to high \acp{SNR} the TSDCE has superior performance, since both methods present similar $P_D$, having the TSDCE lower \acp{RMSE}. Regarding the codebook sizes, both methods provide enhanced \ac{DoA} estimations for $Q=P=32$, as expected from previous \ac{NMSE} comparisons.

\begin{figure*}[t!]
\begin{center}
\includegraphics[width =0.7\textwidth]{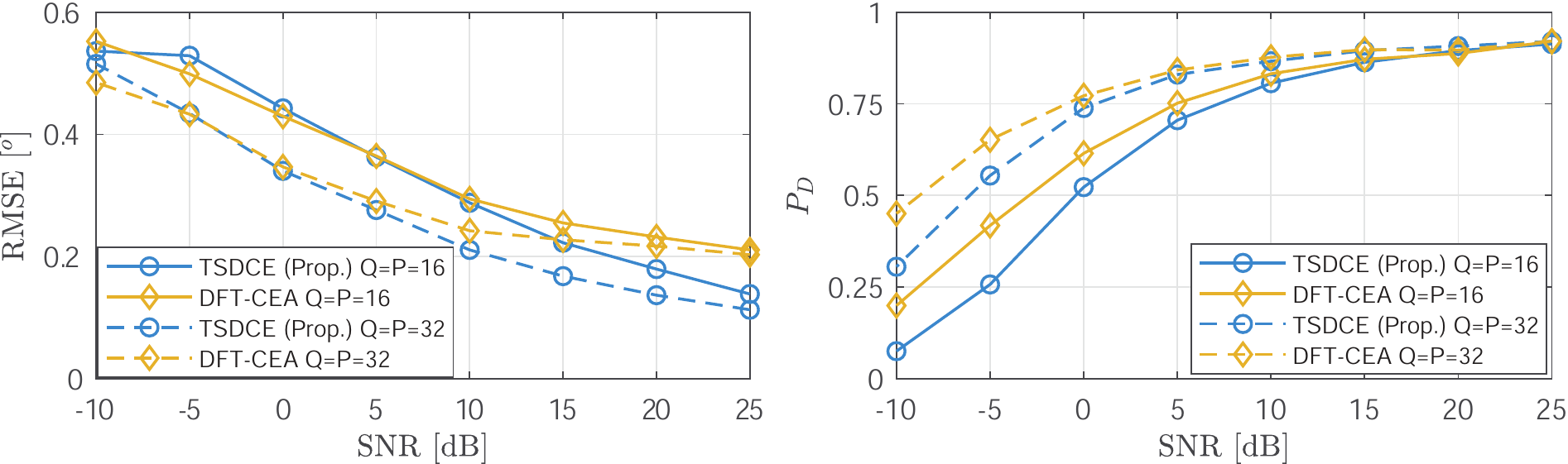}
\caption{{\ac{DoA} \ac{RMSE} and $P_D$ of channel estimation with TSDCE and DFE-CEA for different values of \ac{SNR}, $L=3$ and $n_t=n_r=16$.}}\label{fig:resultsDOA} 
\end{center}
\end{figure*}

Fig.~\ref{fig:resultsL} allows to see the effect of increasing the number of channel paths and antennas, setting $K=3$ for the TSDCE. According to Fig.~\ref{fig:resultsL}~(a), with $L=6$ and $n_t=n_r=32$ the TSDCE is still the method with the best performance at medium to high \acp{SNR}, however, it is not able to reach the CRLB due to the increased number of paths. This result was expected, due to the estimation of a less sparse channel than in the $L=3$ case. In Fig.~\ref{fig:resultsL}~(b), where $n_t=n_r=64$, there is a wider peformance gap with respect to the CRLB, and the performance of the DFT-CEA is severely affected by the beamwidth reduction. In Fig.~\ref{fig:resultsL}~(c) it can be seen that a further increase of the number of paths ($L=8$) penalizes the TSDCE, but it still provides the best performance for positive SNRs.

\begin{figure*}[t!]
\begin{center}
\includegraphics[width =0.85\textwidth]{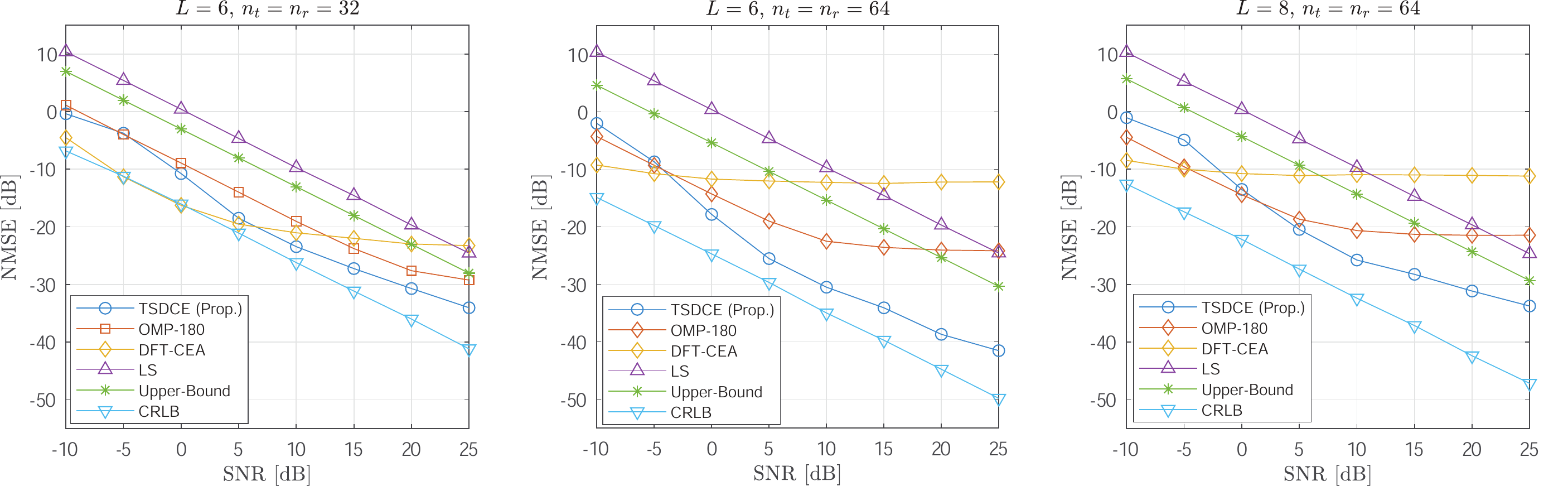}
\caption{{\ac{NMSE} of channel estimation for different values of \ac{SNR} and various numbers of paths and antennas, with codebook sizes matching the number of antennas.}}\label{fig:resultsL} 
\end{center}
\end{figure*}

Fig.~\ref{fig:results4} shows a performance comparison with the \ac{KF}-based beamtracking method proposed in \cite{Zha16} for $n_t = n_r = 16$ and $Q=P=16$. For this tracking scenario, we considered 500 simulated blocks with 100 time slots per block. The \ac{AoA}s and \ac{AoD}s in time slot $j$ are obtained by adding the realization of a zero-mean Gaussian random variable with variance $\sigma^{2}_u$ to the angles in time slot $j-1$. Channel coefficients remain unchanged within the same block. For the \ac{KF} method, the true channel parameters are used as initial tracking estimates for each block. Moreover, we assume perfect knowledge of the channel complex coefficient for the \ac{KF}-based method, which only re-estimates the \ac{AoA} and \ac{AoD} angles at each slot. The algorithms are evaluated considering three different values of the angle standard deviation, $\sigma_u \in \left\{0.5^\circ, 1^\circ, 1.5^\circ \right\}$. A single performance curve is displayed for our proposal, since it is invariant with $\sigma_u$. At $\sigma_u=0.5^\circ$, the \ac{KF} method is only outperformed by TSDCE at high \ac{SNR}. However, it can be seen an increasing degradation of the \ac{KF}-based beam tracking approach for higher angle deviations. The crossing point appears at much lower \ac{SNR}s when $\sigma_u=1^\circ$ or $1.5^\circ$. In fact, at $\sigma_u=1.5^\circ$ the high \ac{NMSE} values make the beam tracking nearly independent from the \ac{SNR}.

\begin{figure}
\begin{center}
\includegraphics[width =0.7\columnwidth]{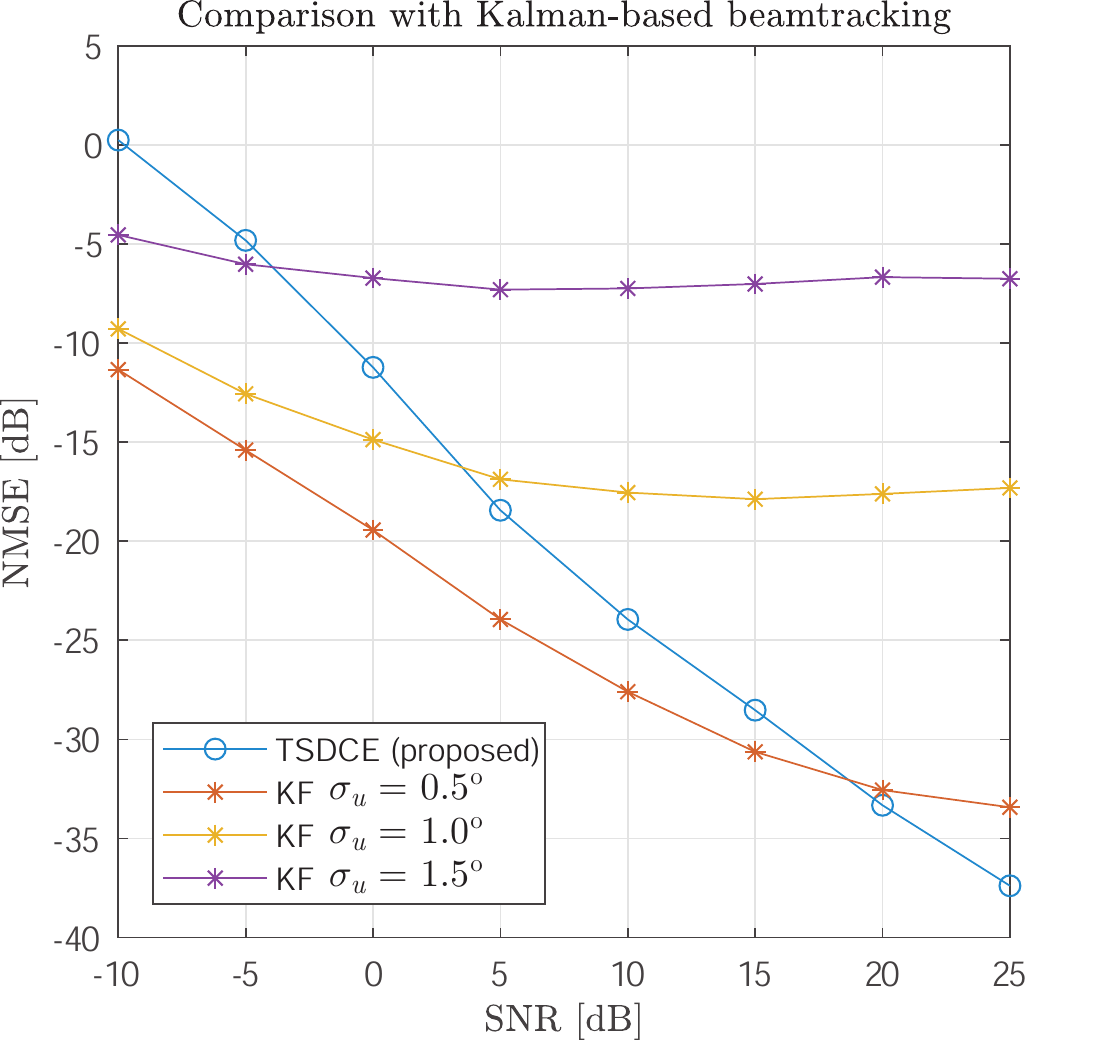}
\caption{{Performance comparison (with $L=3$) w.r.t. \ac{KF} beam tracking \cite{Zha16} at different angle deviations, with $n_t = n_r = 16$ and $Q=P=16$.}}\label{fig:results4} 
\end{center}
\end{figure}

The complexity order of the TSDCE is analyzed below based on the complexities of its core operations provided in \cite{Golub96}. The main computational complexity of Algorithm~\ref{Alg:Alg1} comes from the $\mathrm{IDFT_{2D}}$ in line 1, in the order of $\mathcal{O}(PQ \log_2 (\min (P,Q)))$, as well as from the SVD in line 7, which is $\mathcal{O}((L-1)\cdot n_tn_r \min (n_t,n_r))$, and the $\mathrm{ACF_{2D}}$ in line 12 (carried out for $K$ rounds), in the order of $\mathcal{O}((KL) \cdot 2n_tn_r \log_2 (\min (n_t,n_r)))$. Note that the SVD-related complexity term is the dominant one for a large number of antennas. Table~\ref{tab:complexity} compares the complexity order of the proposed approach with the ones of the LS and OMP methods (derived in~\cite{lee2014}) and of the DFT-CEA method \cite{Mon15}, in the specific case where $QP=n_tn_r$. For instance, when $n_t=n_r=64$ with $L=3$ and the parameters used in the simulations, the complexity orders are TSDCE: $\mathcal{O}(5.2\cdot 10^5)$; LS: $\mathcal{O}(7\cdot 10^{10})$; OMP-180: $\mathcal{O}(8\cdot 10^9)$ and DFT-CEA: $\mathcal{O}(1\cdot 10^7)$. The analysis shows that the proposed TSDCE algorithm is clearly the fastest method in the comparison. In addition, efficient implementations such as the ones derived for the SVD in \cite{Hol08} and \cite{Men11}, and for the IFFT in \cite{Lin05}, can be readily applied to the TSDCE algorithm, achieving a further reduction in the algorithm complexity.

\begin{table*}[t]
\begin{center}
\caption{{Complexity of different channel estimation algorithms for $QP=n_tn_r$}}
\begin{tabular}{c|c|c|c}
\cline{1-4}
\textbf{TSDCE}                                            & \textbf{LS}                       & \textbf{OMP}                & \textbf{DFT-CEA}                    \\ \cline{1-4}
$\mathcal{O}((L-1) \cdot n_tn_r \min (n_t,n_r))$ & $\mathcal{O}((n_tn_r)^3)$ & $\mathcal{O}(G^2 Ln_tn_rI_t)$ & $\mathcal{O}(N^2_{DFT}\log_2N_{DFT})$ \\ \cline{1-4}
\end{tabular}\label{tab:complexity}
\end{center}
\end{table*}

\section{Conclusion}\label{Sec:Conc}
In this paper, a low-complexity algorithm for channel estimation in the transformed spatial domain has been proposed considering analog mmWave systems with AWGN. The key step of the method relies on transforming the directional angular observations into the transformed spatial domain, using rank-one approximations and the sample autocorrelation function as pre-processing steps for robust path extraction. Then, the parameters of each path are estimated by following a frequency estimation approach based on weighted least squares. Analytical upper and lower bounds have been stated, showing that the normalized mean square error of the proposed algorithm remains within such bounds and close to the Cramer-Rao lower bound.
The algorithm has been compared with widely used solutions such as least squares, orthogonal matching pursuit and DFT-based channel estimation, showing that it outperforms the benchmarks at a remarkably lower computational complexity, and without saturating at high signal-to-noise ratios. The comparison with Kalman-based beamtracking shows that the proposed approach is invariant with respect to the angle standard deviation. 
Indeed, the spatial domain interpretation of the channel estimation problem is highly effective for achieving both low-complexity and accurate channel parameter estimates. Future work will consider the extension of the method to planar antenna arrays, hybrid beamforming mmWave architectures and mmWave systems using reconfigurable intelligent surfaces. Additionally, further complexity reduction of the proposed method through fast SVD implementations will be investigated for systems with a large number of antennas. 

\ifCLASSOPTIONcaptionsoff
  \newpage
\fi

\bibliographystyle{IEEEtran}
\bibliography{Beamtrack}

\begin{IEEEbiography}[{\includegraphics[width=1in,height=1.25in,clip,keepaspectratio]{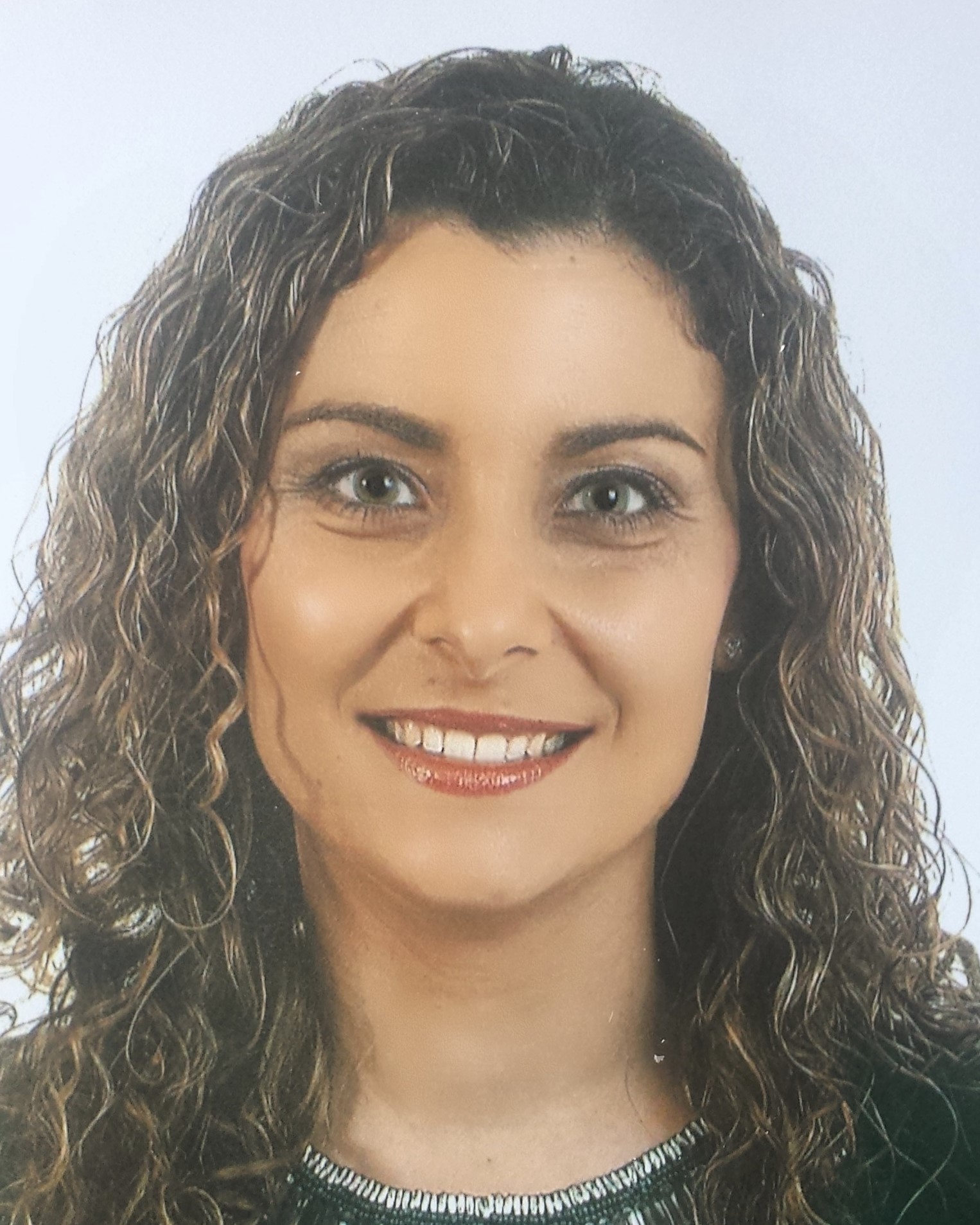}}]{Sandra Roger} (Senior Member, IEEE) received the Ph.D. degree in telecommunications engineering from the Universitat Polit\`ecnica de Val\`encia (UPV), Spain, in 2012. During her doctorate studies, she performed two research stays at the Institute of Telecommunications, Vienna University of Technology, Austria. From July 2012 to December 2018, she was a Senior Researcher with the iTEAM Research Institute, UPV, where she worked in the European projects METIS and METIS-II on 5G design. In January 2019, she joined the Computer Science Department of the Universitat de Val\`encia as a Senior Researcher (“Ramon y Cajal” Fellow). Dr. Roger has authored/coauthored around 60 papers in renowned conferences and journals. Her main research interests are in the field of signal processing for communications, vehicular communications, and wireless system design.
\end{IEEEbiography}

\begin{IEEEbiography}[{\includegraphics[width=1in,height=1.25in,clip,keepaspectratio]{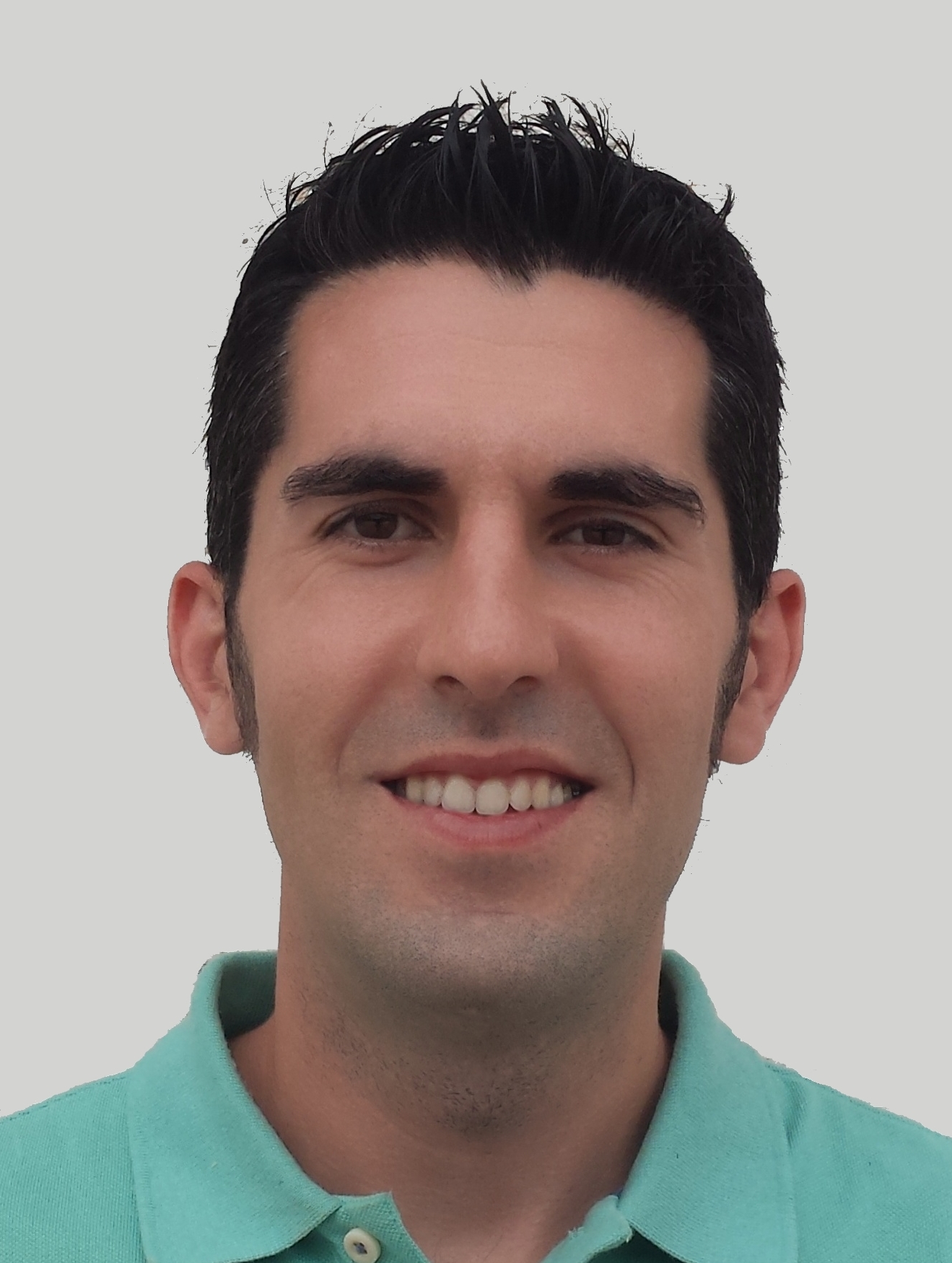}}] {Maximo Cobos} (Senior Member, IEEE) received the master’s degree in telecommunications and the Ph.D. degree in telecommunications engineering from the Universitat Politècnica de València, Spain, in 2007 and 2009, respectively. He completed with honors his studies under University Faculty Training Program (FPU) and was the recipient of the Ericsson Best Ph.D. Thesis Award from Spanish National Telecommunications Engineering Association. In 2010, he received a Campus de Excelencia Postdoctoral Fellowship to work with the iTEAM research institute in Valencia. In 2011, he joined the Universitat de València, where he is currently an Associate Professor. In 2009 and 2011 he was a guest researcher at T-Labs Berlin, Germany, and, in 2019 at Politecnico di Milano, Italy, where he also held an Adjunct Professorship from 2020 to 2021. His work is focused on the area of digital signal processing and machine learning for wireless sensor networks, audio and multimedia applications, where he has authored or coauthored more than 100 technical papers in international journals and conferences. He is a member of the Audio Signal Processing Technical Committee of the European Acoustics Association and serves as an Associate Editor for the IEEE SIGNAL PROCESSING LETTERS.
\end{IEEEbiography}

\begin{IEEEbiography}[{\includegraphics[width=1in,height=1.25in,clip,keepaspectratio]{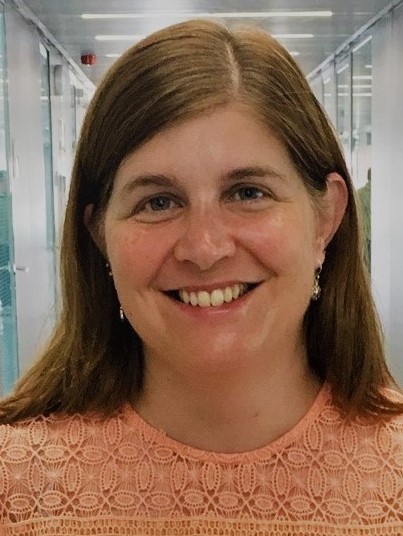}}] {Carmen Botella-Mascarell} (Senior Member, IEEE) received  her  M.Sc.  and Ph.D. degrees in Telecommunications Engineering from Universitat Politècnica de València, Spain, in 2003 and 2008, respectively. In 2009 and 2010, she  was  a  postdoctoral  researcher  in  the Communications Systems and Information Theory  group,  Chalmers  University  of  Technology,  Sweden. In 2011, she joins the Computer Science Department of the Universitat de València where she is currently an Associate Professor. Dr Botella-Mascarell has  authored/coauthored  75  technical  papers  in  international  conferences and journals. Her research interests include the general areas of coordination and cooperation in wireless systems, with special focus on physical-layer solutions for 5G and beyond.
\end{IEEEbiography}

\begin{IEEEbiography}[{\includegraphics[width=1in,height=1.25in,clip,keepaspectratio]{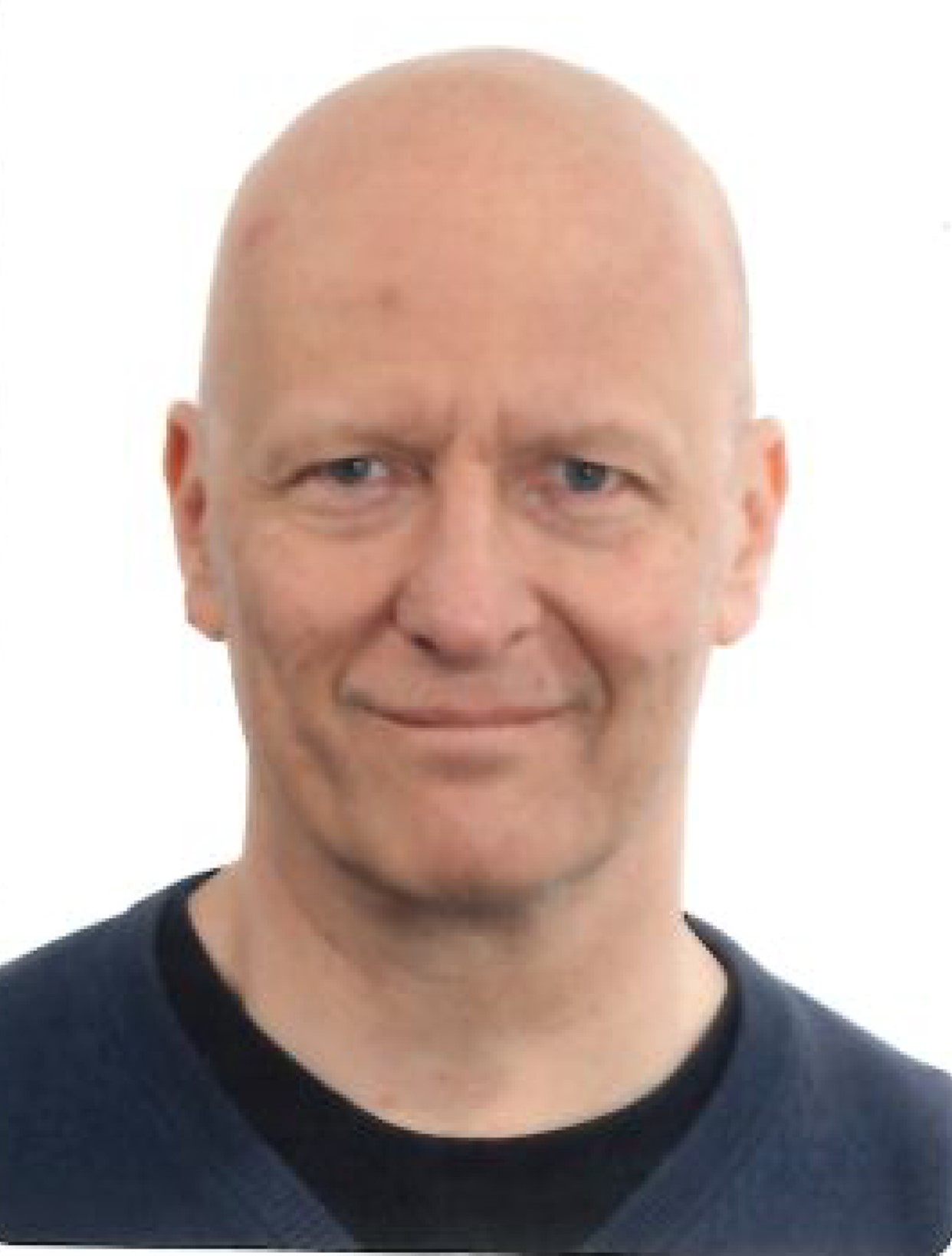}}] {Gábor Fodor} (Senior Member, IEEE) received the Ph.D. degree in electrical engineering from the Budapest University of Technology and Economics in 1998 and the D.Sc. degree from the Hungarian Academy of Sciences (Doctor of MTA) in 2019. He is currently a Master Researcher at Ericsson Research and a Docent and an Adjunct Professor at the KTH Royal Institute of Technology, Stockholm, Sweden. He has authored or coauthored more than 150 refereed journal articles and conference papers and seven book chapters and holds more than 100 granted European and U.S. patents. He was a co-recipient of the IEEE Communications Society Stephen O. Rice Prize in 2018 and the Best Student Conference Paper Award from the IEEE Sweden VT/COM/IT Chapter in 2018. Dr. Fodor is currently the Chair for the IEEE Communications Society Emerging Technology Initiative on Full Duplex Communications. From 2017 to 2020, he was also a member of the Board of the IEEE Sweden joint Communications, Information Theory and Vehicle Technology Chapter. He is currently serving as an Editor for IEEE TRANSACTIONS ON WIRELESS COMMUNICATIONS and IEEE WIRELESS COMMUNICATIONS.
\end{IEEEbiography}
\end{document}